\pdfoutput=1
\RequirePackage{latexml}

\PassOptionsToPackage{
	pagebackref
}{hyperref}

\iflatexml
	\documentclass{article}
	\title{Proportional Aggregation of Preferences for Sequential Decision Making}
	\author{Nikhil Chandak \\ Max Planck Institute for Intelligent Systems \\ Tübingen AI Center
	\and
	Shashwat Goel \\  Max Planck Institute for Intelligent Systems \\ ELLIS Institute T\"ubingen
	\and
	Dominik Peters \\ CNRS, LAMSADE, Universit\'e Paris Dauphine - PSL
	}
\else
	\documentclass[11pt]{scrartcl}
	\setcapindent{0pt}
	
	\usepackage[a4paper, total={16cm, 24cm}]{geometry}
	
	\title{\vspace{-1cm}Proportional Aggregation of Preferences for \\ Sequential Decision Making}
	
	\usepackage{authblk}

	\author[1,2]{Nikhil Chandak}
	\author[1,3]{Shashwat Goel}
	\author[4]{Dominik Peters}
	\affil[1]{Max Planck Institute for Intelligent Systems}
	\affil[2]{Tübingen AI Center}
	\affil[3]{ELLIS Institute T\"ubingen} 
	\affil[4]{CNRS, LAMSADE, Universit\'e Paris Dauphine - PSL}
	\usepackage{wrapstuff}
\fi

\usepackage[pagebackref]{hyperref}

\hypersetup{
	pdfencoding=auto, 
	psdextra,
	colorlinks=true,
	citecolor=green!50!black,
	linkcolor=red!60!black,
	urlcolor=blue!90!black
}
\usepackage{amsmath,amssymb,amsthm}
\usepackage{subcaption}
\usepackage{booktabs}
\usepackage{enumitem}

\usepackage{tikz}
\usetikzlibrary{decorations.pathreplacing,calc}
\newcommand{\xmark}{\tikz[scale=0.23,draw=red!50!black] { \draw[line width=0.7,line cap=round] (0,0) to [bend left=6] (1,1);
    \draw[line width=0.7,line cap=round] (0.2,0.95) to [bend right=3] (0.8,0.05);
}}
\newcommand{\cmark}{\tikz[scale=0.23,draw=green!50!black] {  \draw[line width=0.7,line cap=round] (0.25,0) to [bend left=10] (1,1);
    \draw[line width=0.8,line cap=round] (0,0.35) to [bend right=1] (0.23,0);
}}

\usepackage[nameinlink]{cleveref}
\AddToHook{cmd/appendix/before}{\crefalias{section}{appendix}\crefalias{subsection}{appendix}}

\usepackage[vlined,noend]{algorithm2e}

\theoremstyle{plain}
\newtheorem{theorem}{Theorem}[section]
\newtheorem*{theorem*}{Theorem}
\newtheorem{proposition}[theorem]{Proposition}

\newtheorem{corollary}[theorem]{Corollary}
\theoremstyle{definition}
\newtheorem{definition}[theorem]{Definition}

\newtheorem{remark}[theorem]{Remark}
\Crefname{algocf}{Algorithm}{Algorithms}

\usepackage{natbib}

\usepackage{thmtools}
\usepackage{thm-restate}

\renewcommand*{\le}{\leqslant}
\renewcommand*{\leq}{\leqslant}
\renewcommand*{\ge}{\geqslant}
\renewcommand*{\geq}{\geqslant}
\renewcommand{\epsilon}{\varepsilon}

\newcommand{\alternativea}{\tikz[baseline=-3pt,transform shape,scale=0.85]{\node[text height=6pt,inner sep=1pt,fill=red!20,    circle,minimum width=12pt] (x) {$a$};}}
\newcommand{\alternativeb}{\tikz[baseline=-3pt,transform shape,scale=0.85]{\node[text height=6pt,inner sep=1pt,fill=blue!20,   circle,minimum width=12pt] (x) {$b$};}}
\newcommand{\alternativec}{\tikz[baseline=-3pt,transform shape,scale=0.85]{\node[text height=6pt,inner sep=1pt,fill=green!20,  circle,minimum width=12pt] (x) {$c$};}}
\newcommand{\alternatived}{\tikz[baseline=-3pt,transform shape,scale=0.85]{\node[text height=6pt,inner sep=1pt,fill=orange!20, circle,minimum width=12pt] (x) {$d$};}}
\newcommand{\alternativee}{\tikz[baseline=-3pt,transform shape,scale=0.85]{\node[text height=6pt,inner sep=1pt,fill=magenta!20,circle,minimum width=12pt] (x) {$e$};}}
\newcommand{\alternativef}{\tikz[baseline=-3pt,transform shape,scale=0.85]{\node[text height=6pt,inner sep=1pt,fill=cyan!20,   circle,minimum width=12pt,font=\small] (x) {$f$};}}
\newcommand{\alternativeg}{\tikz[baseline=-3pt,transform shape,scale=0.85]{\node[text height=6pt,inner sep=1pt,fill=brown!20,    circle,minimum width=12pt] (x) {$g$};}}
\newcommand{\alternativeh}{\tikz[baseline=-3pt,transform shape,scale=0.85]{\node[text height=6pt,inner sep=1pt,fill=violet!20,    circle,minimum width=12pt] (x) {$h$};}}

\usepackage[backgroundcolor=orange!15!white,bordercolor=orange!80!black,textsize=small]{todonotes}
\tikzset{/tikz/notestyleraw/.append style={rounded corners=0pt,inner sep=0.6ex}}
\presetkeys{todonotes}{inline}{}

\begin{document}

\iflatexml
\date{March 2025 $\cdot$ arXiv version v2}
\else
\date{\vspace{-5pt}\large\sffamily March 2025 $\cdot$ arXiv version v2\thanks{In this version, in addition to improving and updating the exposition, we have renamed our axioms compared to arXiv version v1 and the version published at AAAI 2024. We have renamed EJR $\to$ Weak EJR and Strong EJR $\to$ EJR, and analogously for PJR and JR, following requests from other authors working on this model.}}
\fi

\maketitle

\vspace{-25pt}
\begin{abstract}
    We study the problem of fair sequential decision making given voter preferences. In each round, a decision rule must choose a decision from a set of alternatives where each voter reports which of these alternatives they approve. Instead of going with the most popular choice in each round, we aim for proportional representation across rounds, using axioms inspired by the multi-winner voting literature. 
    The axioms require that every group of $\alpha\%$ of the voters that agrees in every round (i.e., approves a common alternative), must approve at least $\alpha\%$ of the decisions. A stronger version of the axioms requires that every group of $\alpha\%$ of the voters that agrees in a $\beta$ fraction of rounds must approve $\beta\cdot\alpha\%$ of the decisions. We show that three attractive voting rules satisfy axioms of this style. One of them (Sequential Phragm\'en) makes its decisions online, and the other two satisfy strengthened versions of the axioms but make decisions semi-online (Method of Equal Shares) or fully offline (Proportional Approval Voting). 
     We present empirical results for these rules based on synthetic data and U.S. political elections. We also run experiments using the moral machine dataset about ethical dilemmas: We train preference models on user responses from different countries and let the models cast votes. We find that aggregating these votes using our rules leads to a more equal utility distribution across demographics than making decisions using a single global preference model.
\end{abstract}

\iflatexml
\textbf{Note:} In this version, in addition to improving and updating the exposition, we have renamed our axioms compared to arXiv version v1 and the version published at AAAI 2024. We have renamed EJR $\to$ Weak EJR and Strong EJR $\to$ EJR, and analogously for PJR and JR, following requests from other authors working on this model

\else
\setcounter{tocdepth}{1} 
\renewcommand{\contentsname}{\vspace{-1.7em}}
\DeclareTOCStyleEntry[beforeskip=6pt]{section}{section}
\tableofcontents
\fi

\section{Introduction}
\label{sec:intro}

We consider the problem of making a sequence of independent decisions via voting.
In each round, we can choose one alternative from a set of several alternatives, based on voters who tell us which alternatives they support (or \emph{approve}). The set of voters stays the same across rounds, though the set of alternatives may change. The ``standard'' way of making such decisions is to take the alternative with the most supporters in each round. A problem with this method is that non-majority groups of voters may have very little influence on the outcomes. For example, if there is a fixed group of 51\% of the voters who all report the same opinion in every round, then 100\% of the decisions will be taken in accordance with the wishes of that group, with the other 49\% of voters ignored or at most acting as tie-breakers. In many contexts, this is undesirable due to fairness concerns. 

Following recent work that studies this model under the name ``perpetual voting'' \cite{lackner2020perpetual}, we define formal properties of voting rules that capture the intuition that a group of $\alpha\%$ of the voters should be able to control the decisions in $\alpha\%$ of the rounds. Inspired from work in multi-winner approval voting \citep{lacknerskowron}, we then define a number of voting rules that satisfy these properties. We believe that these rules have promising applications in a variety of domains. Here are some examples:
\begin{itemize}
	\item \emph{Hiring decisions}: Consider a department that hires new faculty each year, with existing faculty voting over the applicants. The department wishes to hire people representing its spectrum of research interests. For example, if 20\% of the department works on one topic and votes for candidate on that topic, then at least 1 such candidate should be hired every 5 years. 
	\item \emph{Virtual democracy}: In cases where a group of people need to make an extremely large number of decisions, we may wish to automate this process. An approach known as virtual democracy does this by initially learning voters' preferences over a space of potential alternatives (specified by feature vectors) for example based on pairwise comparisons. Then, each decision is made by letting the models vote on the decision maker's behalf by predicting their preferences. This approach has led to proof-of-concept systems that automate moral decisions faced by autonomous vehicles \citep{noothigattu2018voting}, kidney exchanges \citep{freedman2020adapting}, collective decision making directly from natural language preferences \citep{mohsin2021making} and allocation of food donations \citep{lee2019webuildai}. The approach behind these systems has recently been criticized as overweighting the opinions of majorities \citep{moralmachinetyranny}. We show preliminary evidence in our experiments that using proportional voting methods could avoid this issue.
	\item \emph{Combining the outputs of sequence models}: Modern large language models (LLMs) operate in rounds, in each round suggesting a probability distribution over tokens. We could view these distributions as preferences over tokens, and use proportional methods to mix several LLMs together, potentially allowing for fine-grained customization to user preferences, or to aid in drafting compromise documents. \citet{peters2024frontiers} discusses this use case in detail.
	\item \emph{Policy decisions of coalition governments}: In many countries, the government is formed by coalitions of several parties with different strengths (for example, in Germany 2021--25 it consists of 3 parties who received 26\%, 15\%, and 12\% of the vote). A coalition needs to agree on a program, consisting of decisions on many issues. Our methods could help to design a program where each party's preferences gets a fair representation in the outcome.
	\item \emph{The parliamentary approach to moral uncertainty}: An agent may have moral uncertainty, and be unsure which of several moral theories is right. \citet{bostrom2009} has suggested that an agent could decide on what actions to take given this uncertainty by imagining how the decision would be made by a \emph{parliament} in which each moral theory is represented according to the agent's credence in the theory. The idea is that the representatives in the imaginary parliament would deliberate, bargain, and vote, and thereby reach good compromise decisions \citep{newberry2021parliamentary}. While such bargaining is difficult to make precise \citep[but see][]{greaves2023bargaining}, proportional multi-issue voting rules could simulate this process and provide one way of making such decisions.
\end{itemize}

In many of these examples, the different issues are not completely separable, so our model may not be a precise fit. However, when using online rules, the agents can adjust their reported preferences to take some consistency properties into account.

Other applications that are mentioned in the literature are the elections to a committee, where only part of the committee is replaced in each election -- thereby inducing a temporal element \citep{elkind2024survey}, or a company that donates some resources to a charity each year and lets employees or customers voter over which charity to choose \citep{elkind2024verifying}, or meal selection by a group of friends \citep{elkind2024temporal}.

\paragraph{Our Results}

\begin{table*}
	\centering
	\newcommand{\weakejrremark}{\hspace{1pt}\raisebox{0pt}{\textsuperscript{+}}\hspace{-6pt}}
	\newcommand{\ejrremark}{\hspace{1pt}\raisebox{0pt}{\textsuperscript{\textdagger}}\hspace{-6pt}}
	\begin{tabular}{lcccc}
		\toprule
		Voting Rule & \:Weak PJR\: & PJR & \:\:Weak EJR\:\: & EJR \\
		\midrule
		Sequential Phragm\'en (Online) & \cmark & \cmark & \xmark\weakejrremark{} & \xmark\ejrremark{} \\
		Method of Equal Shares (Semi-online) & \cmark &  \xmark & \cmark& \xmark\ejrremark{} \\
		PAV (Offline) & \cmark & \cmark & \cmark & \cmark \\
		\bottomrule
	\end{tabular}
	\caption{Our results. For each setting (online, semi-online (where the rule only knows the total number of rounds in advance), offline (where the rule knows all the preferences in advance)), we list the best known rule with respect to proportionality axioms. $^{+}$ = knowing an online rule satisfying Weak EJR would resolve a major open problem in the theory of multi-winner voting. \textsuperscript{\textdagger} = no semi-online rule can satisfy EJR due to a result of \citet{elkind2024verifying}.}
	\label{tab:results}
\end{table*}

We build on the work of \citet{bulteau2021jr} who defined two axioms that guarantee group representation. We call them PJR (Proportional Justified Representation) and Weak PJR.
These axioms do not require the groups to be pre-defined. Instead, the guarantee applies to \emph{all} subsets of voters who support a common alternative in some of the rounds, and the ``weak'' versions only apply if they support a common alternative in \emph{all} of the rounds. \citet{bulteau2021jr} conjectured that no polynomial-time algorithm can find a Weak PJR outcome. We show that the picture is more positive: an attractive polynomial-time voting rule building on ideas of \citet{Phra94a} satisfies PJR (\Cref{sec:guarantees:phragmen}). This rule is fully online, making decisions round by round.

We then define axioms called Weak EJR (Extended Justified Representation), which provides better group representation guarantees than Weak PJR, and analogously EJR. 
We show that the existence of a fully online rule that satisfies Weak EJR would positively resolve a major open problem in multi-winner approval voting \citep{lacknerskowron}. 
Still, we can define a simple voting rule (a variant of the ``Method of Equal Shares'') that satisfies Weak EJR and that is ``semi-online'': it needs to know how many rounds there will be in total, but it can be given preferences in an online fashion (\Cref{sec:guarantees:mes}). However, it does not satisfy PJR or EJR. 
This suggests the question whether we can satisfy EJR in the offline setting, where all preferences are available in advance. In \Cref{sec:guarantees:pav}, we show that Proportional Approval Voting (PAV), an offline rule that maximizes a carefully chosen objective function proposed by \citet{Thie95a}, does satisfy all our axioms, including EJR. While the PAV optimization problem is NP-hard, we prove these results also hold for a polynomial-time local search variant of PAV. Finally, in \Cref{sec:hardness} we show that further strengthenings of the PJR/EJR axioms are not satisfiable, and that on inputs where a solution exists, online rules still must fail them. 

We close with several experimental results on synthetic (\Cref{sec:experiments:synthetic}) and real datasets (\Cref{sec:experiments:political}), on which we compare our rules to existing rules from the literature (which do not satisfy our axioms). 
In \Cref{sec:experiments:learning}, we study the \emph{virtual democracy} application based on the \emph{moral machine} dataset \citep{awad2018moral}. We learn separate preference models that predict the respondents' judgement for each country and use each country's model as a `voter'. We then sample new decision situations and query the models for their top choices among the alternatives in each situation, aggregating the responses using our voting rules to produce a decision sequence. We compare the decisions produced with training a single universal model on equal number of responses from all countries, and taking its top choices as the decisions. We find that when voters do not already have very similar preferences on the issues, the aggregation approach leads to much fairer outcomes than the decisions made by the combined model.

 \section{Related Work}
\label{sec:related}

\paragraph{Perpetual Voting}
Our work is closely related to perpetual voting \citep{lackner2020perpetual} which studies online rules, while we also consider offline rules. \citet{lackner2020perpetual} focused on \emph{individual} fairness properties, for example requiring that each voter approves at least 1 decision in every time interval of some bounded length. We are interested in guarantees for \emph{groups} of voters who agree with each other, such that larger groups receive stronger guarantees. Such guarantees have been considered by \citet{lacknermaly2023} and \citet{bulteau2021jr}; we compare them to ours in \Cref{sec:formulation}. \citet{kozachinskiy2024optimalbounds} study the achievable egalitarian welfare in perpetual voting.

\paragraph{Public Decision Making}
\citet{fairpublicdecisionmaking} study ``public decision making'' which is an offline model of several decisions, where a voter's total utility is the sum of the utilities obtained in each round. We study the special case where the utility values in each round are restricted to 0 and 1. \citet{fairpublicdecisionmaking} focus on fairness notions for individuals (not groups) derived from fair division. 
\citet{freeman2020proportionality}, \citet{skowron2022proportionalpublicdecisions}, and \citet{propinterdependentissues} consider the special case where there are exactly two alternatives (``yes/no'') in each round and utilities are 0/1, which is a special case of our setting where fairness properties are easier to obtain. 

\paragraph{Repeated Portioning}
In our model, we need to choose one alternative in each round. One could define a fractional version of this model, where in each round, we choose a \emph{mixture} of the available alternatives (formally, a probability distribution or lottery over $C_j$). In the single-round setting, finding a fair mixture has been considered by several papers \citep{BMS05a,FGM16a,ABM20a,BBPS21a}. For a multi-round setting with online decisions, this model is studied by \citet[Section 5]{banerjee2023proportionally} who consider a proportional fairness objective and propose methods with good regret bounds. 

\paragraph{Multi-winner Approval Voting} 
For the task of selecting a committee of exactly $k$ out of $m$ candidates given approval votes \citep{lacknerskowron}, proportionality has been intensely studied. Many axioms have been proposed that formalize the idea that $\alpha\%$ of the voters should be able to elect at least $\alpha\% \cdot k$ of the committee members \citep{ejr,pjr,PeSk20,ejr+}. Voting rules proposed long ago by \citet{Thie95a} and \citet{Phra94a} do very well with respect to these axioms \citep{Jans16a,brill2023phragmen}. A more recent rule proposed by \citet{PeSk20} known as the Method of Equal Shares combines some of the advantages of the Phragmén and Thiele rules. \citet{skowron2017proportional} discuss adapting multi-winner rules to obtain ``proportional rankings'' of alternatives, which can be seen as a sequential decision making problem. Many of the rules and axioms we consider in this paper are analogues or generalizations of proposals from multi-winner voting. 

\paragraph{Approval-based Apportionment}
\citet{approvalbasedapportionment} considered a model where we need to assign $k$ seats to political parties, based on voters who submit approval sets over parties. This can be seen as a kind of multi-winner election where candidates can be added to the committee several times (because a single party can receive multiple seats). Alternatively, it can be seen as a special case of our model with $k$ rounds, where the sets of alternatives and all voters' approval sets remain the same across all rounds. Being a special case of our model, some negative results of \citet{approvalbasedapportionment} carry over to our model, notably that Sequential Phragm\'en fails Weak EJR and that PAV is NP-hard to compute; we will mention these connections during our discussion.

\paragraph{Combinatorial Voting}
A classic literature on voting in combinatorial domains \citep{LaXi15a} studies the problem of making decisions on several issues. The main focus is on the representation of complex preferences, and the computational problem of finding good outcomes given such preferences, often in the sense of maximizing utilitarian or egalitarian welfare \citep{amanatidis2015multiple}. Some of this work considers (conditional) approval preferences \citep{barrot2016conditional,barrot2017manipulation}. \\

\noindent
Additional related work can be found in the recent survey by \citet{elkind2024survey}.

 \section{Problem Formulation}
\label{sec:formulation}

\subsection{Model}

Let $N = \{1, 2, \ldots, n\}$ be the set of \emph{voters}, with $|N|=n$. There is a set $R = \{1, 2, \dots, T \}$ of $T$ \emph{rounds}, where $T$ is the \emph{time horizon}. 
In each round $j\in R$, we are given a set of \emph{alternatives} $C_j$. 
Each voter $i \in N$ \emph{approves} some (usually non-empty) subset $\smash{A^i_j}$ of $C_j$. 
Thus, we have a sequence of $T$ decision rounds together with a sequence $C= (C_j)_{j \in R}$ of alternative sets and a collection $A = (A^i_j)_{i\in N, {j \in R}}$ of approval sets.
We let $(N, R, C, A)$ denote a \emph{decision instance}.

A \emph{decision sequence} $D = (d_1, \dots, d_T) \in C_1 \times \dots \times C_T$ specifies a single \emph{decision} $d_j \in C_j$ for each round $j \in R$. For a voter $i \in N$, we write $\smash{U^i_D = |\{ j \in R : d_j \in A^i_j \}|}$ for the number of decisions in $D$ that $i$ approves; we treat $U^i_D$ as $i$'s \emph{utility}. 
A \emph{decision rule} $f$ takes as input a decision instance and returns a decisions sequence $D = f(N, R, C, A)$. We call a decision rule \emph{semi-online} if its decision in round $j \in R$ only depends on the time horizon $T$ and on information up to round $j$, i.e., on $N$, $R$, $C_1, \dots, C_j$ and $\smash{A_1^i, \dots, A_j^i}$. A rule is \emph{online} if in addition, the decisions are independent of the time horizon $T$.

\subsection{Axioms}
\label{sec:axioms}

\begin{figure}[t]
	\centering
	\scalebox{0.95}{\begin{tikzpicture}[yscale=0.96,xscale=1.5]
			\node at (0,5) (strongEJR) {EJR};
			\node at (0,3) (EJR) {Weak EJR};
			
			\node at (2,3.9) (strongPJR) {PJR};
			\node at (2,2) (PJR) {Weak PJR};
			
			\node at (4,2.9) (strongJR) {JR};
			\node at (4,1) (JR) {Weak JR};
			
			\node[align=center] at (6,2.5) (price) {perpetual\\priceability};
			\node[align=center, inner sep=1pt] at (6,0.5) (plq) {perpetual \\ lower quota \\ for closed groups};
			\node[align=center, inner sep=3pt] at (6,-1.5) (lq) {lower quota \\ for closed groups};
			
			\draw[->] 
			(strongEJR) edge (EJR) 
			edge (strongPJR)
			(strongPJR) edge (strongJR)
			edge (PJR)
			(strongJR) edge (JR)
			(EJR)		edge (PJR)
			(price)		edge (plq)
			(plq)         edge (lq)
			(PJR)		edge[bend right=5] (lq)
			edge (JR);
			
			\fill[white] (3.5, 2.4) rectangle (4.5, 2.36); 
			
			\draw[->] (price) edge[bend right=5] (PJR);
	\end{tikzpicture}}
	\caption{Implications between axioms. Perpetual priceability and (perpetual) lower quota for closed groups are discussed in \Cref{app:moreaxioms}.}
	\label{fig:hasse}
\end{figure}

We now define several properties (or \emph{axioms}) of decision sequences that formalize the idea of proportional representation. These come in different strengths and \Cref{fig:hasse} shows implication relationships between them. We will focus on the four properties of (Weak) EJR and (Weak) PJR, shown in bold. For completeness, the figure also mentions two properties introduced by \citet{lacknermaly2023}
(perpetual priceability and lower quota compliance). We discuss these properties, as well as Pareto-efficiency, in \Cref{app:moreaxioms}.

Our first two axioms were introduced by \citet{bulteau2021jr}.\footnote{\label{footnote:pjr} We use different names for these axioms. \citet{bulteau2021jr} say ``all periods intersection PJR'' for what we call Weak PJR, and say ``some periods intersection PJR'' for what we call PJR.}
A group $S\subseteq N$ of voters \textit{agrees} in round $j \in R$ if there is an alternative $c \in C_j$ that all voters in $S$ approve, so $\bigcap_{i\in S} A^i_j \neq \emptyset$. Weak PJR (Weak Proportional Justified Representation, first defined for multi-winner voting by \citealp{pjr}) requires that a group of an $\alpha$ fraction of the voters, if it agrees in every round, must be ``happy'' in at least $\lfloor \alpha T \rfloor$ rounds, meaning that at least one member of $S$ approves the decision (but this member of $S$ can differ across rounds).

\begin{definition}[Weak PJR]
    \label{def:PJR}
    A decision sequence $D$ satisfies \emph{Weak PJR} if for every $\ell \in \mathbb N$ and every group of voters $S \subseteq N$ that agrees in every round and has size $|S| \geq \ell \cdot \frac{n}{T}$, there are at least $\ell$ rounds $j \in R$ in which the decision $d_j$ of $D$ is approved by some voter in $S$ (i.e., $d_j \in \bigcup_{i \in S} \smash{A_j^i}$).
\end{definition}

Weak PJR only provides guarantees for groups that agree in \emph{all} rounds. PJR also gives guarantees for groups that agree only in some of the rounds, though the groups need to be larger: a group that agrees in $k$ rounds deserves to be happy with $\ell$ decisions if the size of the group is at least $\ell \cdot \frac{n}{k}$ (as compared to $\ell \cdot \frac{n}{T}$ for groups that agree in all rounds).% 
\footnote{Strengthening this axiom to only require the group size to be at least $\ell \cdot \frac{n}{T}$ gives rise to an axiom that may be impossible to satisfy, see \Cref{cor:multi-winner-pjr} in the appendix.}
Note that PJR implies Weak PJR (take $k = T$).

\begin{definition}[PJR]
    \label{def:strongPJR}
    A decision sequence $D$ satisfies \emph{PJR} if for every $k \le T$, every $\ell \in \mathbb N$ and every group of voters $S \subseteq N$ that agrees in $k$ rounds and has size $|S| \geq \ell \cdot \frac{n}{k}$, there are at least $\ell$ rounds $j \in R$ in which the decision $d_j$ of $D$ is approved by some voter in $S$ (i.e., $d_j \in \bigcup_{i \in S} \smash{A_j^i}$).
\end{definition}

An equivalent way of stating this axiom is that an $\alpha$ fraction of the voters who agree in a $\beta$ fraction of the rounds need to be ``happy'' with at least an $\lfloor \alpha \cdot \beta \rfloor$ fraction of the decisions. For example, consider a group $S\subseteq N$ of voters with fixed size $\ell \cdot \frac{n}{T}$, and let us ask what PJR guarantees for this group. If $S$ agrees in all rounds, then it says that $S$ should be happy with $\ell$ decisions. If $S$ agrees in $T/2$ rounds, then it says that $S$ should be happy with $\smash{\lfloor \frac{\ell}{2}\rfloor}$ decisions. Thus, the guarantee provided to a group is proportional to the number of rounds in which the group agrees. Yet another way to think about this axiom is to pretend that only the rounds in which $S$ agrees exist, and then apply Weak PJR to that situation, while allowing the group to be satisfied also by decisions in rounds where it doesn't agree.

\begin{figure}
	\centering
	\begin{tabular}{lcccccccc}
		\toprule
		Round & 1 \& 2 & 3 \& 4 & 5 \& 6 & 7 \& 8 \\
		\midrule
		Voter 1 & $\{\alternativea, \alternativeb\}$ & $\{\alternativea, \alternativeb\}$ & $\{\alternativea, \alternativeb\}$ & $\{\alternativea, \alternativeb\}$ \\
		Voter 2 & $\{\alternativea, \alternativec\}$ & $\{\alternativea, \alternativec\}$ & $\{\alternativea, \alternativec\}$ & $\{\alternativea, \alternativec\}$ \\
		Voter 3 & $\{\alternatived\}$ & $\{\alternatived\}$ & $\{\alternatived\}$ & $\{\alternativee\}$ \\
		Voter 4 & $\{\alternatived\}$ & $\{\alternatived\}$ & $\{\alternatived\}$ & $\{\alternativef\}$ \\
		\bottomrule
	\end{tabular}
	\caption{Example illustrating our axioms.}
	\label{fig:axiom-example}
\end{figure}

To illustrate the axioms, consider the instance shown in \Cref{fig:axiom-example}, with $T = 8$ rounds and 4 voters. The group $S = \{1,2\}$ agrees in all rounds (always approving $\alternativea$). Thus, to satisfy Weak PJR, in at least $\ell = 4$ rounds the outcome needs to be either $\alternativea$, $\alternativeb$, or $\alternativec$ (because $|S| \ge \ell \cdot \frac48$). The group $S' = \{3, 4\}$ agrees in the first 6 rounds (approving $\alternatived$), so with $\ell' = 3$, because $|S'| \ge \ell' \cdot \frac{4}{6}$, PJR requires that in at least $3$ rounds, the outcome is either $\alternatived$, $\alternativee$, or $\alternativef$.

A weakness of Weak PJR and PJR is in how they define $S$ being ``happy'' with a decision (``at least one member of $S$ approves the decision''). This definition can be satisfied by a decision sequence that gives each member of $S$ a utility that is much lower than $\ell$ \citep[Sec.~4.2]{PeSk20}. In the example of \Cref{fig:axiom-example}, the decision sequence $(\alternativeb,\alternativeb,\alternativec,\alternativec,\alternatived,\alternatived,\alternativee,\alternativef)$ satisfies PJR, but the first two voters each only approve the decision in 2 rounds, instead of in 4 rounds. 
Following \citet{ejr}, we can fix this by defining the axioms EJR (Extended Justified Representation) and Weak EJR which require that at least one member $i$ of $S$ must approve at least $\ell$ of the decisions, i.e., must have utility $\smash{U^i_D} \ge \ell$. In the example, the decision sequence $(\alternatived,\alternatived,\alternatived,\alternatived,\alternativea,\alternativea,\alternativea,\alternativea)$ satisfies EJR.

\begin{definition}[Weak EJR]
    \label{def:EJR}
    A decision sequence $D$ satisfies \emph{Weak EJR} if for every $\ell \in \mathbb N$ and every group of voters $S \subseteq N$ that agrees in all rounds and that has size $|S| \geq \ell \cdot \frac{n}{T}$, there is a voter $i\in S$ who approves at least $\ell$ decisions in $D$, i.e., $\smash{U^i_D} \ge \ell$.
\end{definition}

\begin{definition}[EJR]
    \label{def:strongEJR}
    A decision sequence $D$ satisfies \emph{EJR} if for every $k \le T$, every $\ell \in \mathbb N$ and every group of voters $S \subseteq N$ that agrees in $k$ rounds and has size $|S| \geq \ell \cdot \frac{n}{k}$, there is a voter $i \in S$ who approves at least $\ell$ decisions in $D$, i.e., $\smash{U^i_D} \ge \ell$.
\end{definition}

In analogy to the multiwinner literature, we can also define the notions of JR and Weak JR, which are the special cases of the PJR/EJR notions with $\ell = 1$. These axioms are also discussed by \citet{bulteau2021jr}.

\begin{definition}[Weak JR]
	\label{def:JR}
    A decision sequence $D$ satisfies \emph{Weak JR} if for every $S \subseteq N$ with $|S| \geq \frac{n}{T}$ that agrees in all rounds, there is a voter $i\in S$ with $\smash{U^i_D} \ge 1$.
\end{definition}

\begin{definition}[JR]
	\label{def:strongJR}
    A decision sequence $D$ satisfies \emph{JR} if  for every $k \le T$ and every  $S \subseteq N$ with $|S| \geq \frac{n}{k}$ that agrees in $k$ rounds, there is a voter $i \in S$ with $\smash{U^i_D} \ge 1$.
\end{definition}

These notions do not guarantee \emph{proportional} representation since larger groups do not obtain higher guarantees, but they do guarantee some minimal representation.

\citet{elkind2024verifying} studied the computational complexity of verifying whether a given decision sequence satisfies a given axiom. They found that this problem is coNP-complete for each of the six axioms, except that verifying Weak JR is tractable for binary issues, i.e., when in each round there are only two alternatives.

A decision rule $f$ \emph{satisfies} one of these axioms if for all possible inputs, the decision sequence selected by $f$ satisfies it. Note that if $f$ is also online, this means that the proportionality guarantee thereby not only holds for the entire decision sequence, but also for every prefix of it.

\subsection{Methods}

We now define three decision rules that are natural analogues of rules that were first 
proposed for multi-winner elections: the Sequential Phragm\'en rule, the Method of Equal Shares, and Proportional Approval Voting.

\subsubsection{Sequential Phragm\'en}
\citet{Phra94a} proposed an approval-based voting method for electing members of the Swedish parliament. \citet{lacknermaly2023} adapted this rule to the context of perpetual voting, calling their adaptation ``perpetual Phragm\'en''. We follow their definition. The rule makes decisions round by round. Each decision provides a value (or \emph{load}) of 1, which is distributed to voters who approve the decision. In each round, the rule chooses an alternative such that no voter ends up with too much load, thereby prioritizing voters who do not yet agree with many prior decisions. Formally, each voter $i$ starts with load $x^i=0$. Sequential Phragm\'en chooses the alternative for which it can distribute a load of $1$ in a way that minimizes the maximum total load assigned to a voter: at each round $j \in R$, we compute the following value for each alternative $c \in C_j$: 
\[ 
s_c = \min_{S \subseteq \{ i \in N : c \in A^i_j \}}  \frac{\sum_{i \in S} x^i + 1}{|S|}.
\] 

\iflatexml
\begin{figure}
	\centering
	\scalebox{0.9}{\begin{tikzpicture}
			[xscale=0.55, yscale=0.42, font=\small,
			bracepath/.style={decorate, decoration = {brace, mirror, transform={yscale=1.1,yshift=-2pt}}},]
			\foreach 
			\voter/\load/\newload
			in {1/1/3,2/2/2,3/3/1,4/3/1,5/4.5/0,6/6/0,7/6/0}
			{
				\draw (\voter-1,0) rectangle (\voter,\load);
				\ifnum\newload>0
				\draw[blue!40, fill=blue!20] (\voter-1,\load) rectangle (\voter,\load+\newload);
				\fi
				\node at (\voter-0.5,0.4) {$x_\voter$};
			}
			
			\draw[out=90, in=180, ->] (-1.4,3) to (-0.3,4);
			\node[anchor=west] at (-1.8,2.5) {$s_c$};
			
			\draw[bracepath] (3.85,4.05) -- (0,4.05);
			\node[anchor=north] at (2,5.55) {$S$};
			
			\draw[white] (8,0) -- (8.8,0); 
	\end{tikzpicture}}
	\vspace{-15pt}
	\caption{Load distribution}
	\label{fig:load}
\end{figure}
\else
\begin{wrapstuff}[r,type=figure,width=5cm]
	\centering
	\scalebox{0.9}{\begin{tikzpicture}
			[xscale=0.55, yscale=0.42, font=\small,
			bracepath/.style={decorate, decoration = {brace, mirror, transform={yscale=1.1,yshift=-2pt}}},]
			\foreach 
			\voter/\load/\newload
			in {1/1/3,2/2/2,3/3/1,4/3/1,5/4.5/0,6/6/0,7/6/0}
			{
				\draw (\voter-1,0) rectangle (\voter,\load);
				\ifnum\newload>0
				\draw[blue!40, fill=blue!20] (\voter-1,\load) rectangle (\voter,\load+\newload);
				\fi
				\node at (\voter-0.5,0.4) {$x_\voter$};
			}
			
			\draw[out=90, in=180, ->] (-1.4,3) to (-0.3,4);
			\node[anchor=west] at (-1.8,2.5) {$s_c$};
			
			\draw[bracepath] (3.85,4.05) -- (0,4.05);
			\node[anchor=north] at (2,5.55) {$S$};
			
			\draw[white] (8,0) -- (8.8,0); 
	\end{tikzpicture}}
	\vspace{-15pt}
	\caption{Load distribution}
	\label{fig:load}
\end{wrapstuff}
\fi
This value can be understood using a ``water filling'' analogy, as shown in \Cref{fig:load}, where we consider an alternative $c$ approved by 7 voters, and show their current loads as bars. We then fill 1 unit of water on top of the approving voters' loads. Note that the water never falls on top of the loads of voters $5$, $6$, and $7$ because their load is already quite high; in other words, this process has only assigned load to the set $S = \{1,2,3,4\}$. Then $s_c$ is the ``water line'', which is the load of each voter in $S$ after the load of $c$ has been assigned.\footnotemark

\footnotetext{In the multi-winner setting, we can fix $S$ in the definition of $s_c$ to be the set of \emph{all} voters approving $c$, because it cannot happen that an approver of $c$ already has more load than $s_c$, for in that case the rule would have chosen $c$ in an earlier iteration \cite[Lemma 4.5]{brill2023phragmen}. This is not true in our setting because $c$ may not have existed in a prior round. Other definitions of Phragm\'en's method based on virtual bank accounts \cite{Jans16a,PeSk20} do not easily adapt to our setting because of this phenomenon.}

The decision for round $j \in R$ is the alternative $c$ that minimizes $s_c$, breaking ties arbitrarily. After making the decision, we update loads by setting $x^i = s_c$ for each voter $i \in S$ (where $S$ is the coalition attaining the minimum in the definition of $s_c$) and leave $x^i$ unchanged for voters not in $S$. 
Note that the load of a voter will never decrease.

Clearly, Sequential Phragm\'en is an online rule. \citet{lacknermaly2023} illustrate how the rule works using a detailed example and show that it can be computed in polynomial time.

\subsubsection{Method of Equal Shares (MES)}

MES is a recently introduced rule for multi-winner voting \citep{PeSk20} and also used in practice for participatory budgeting \citep{peters2021proportional}. It can be adapted to our setting in a semi-online fashion: the rule needs to know the total number of rounds $T$ in advance, but does not need to know voter preferences of future rounds. Each decision costs $p=\frac{n}{T}$ units, which must be paid by voters that approve the chosen alternative. MES works by subtracting this amount from an initial budget $b_i = 1$ assigned to each voter $i$.

For $\rho \ge 0$, an alternative $c \in C_j$ is called $\rho$-\emph{affordable} if 
\[
\textstyle
\sum_{i \in N : c \in A^i_j} \min(b_i, \rho) \geq p,
\] 
so approvers of $c$ can pay $p$ with no one paying more than $\rho$. 
In each round, the alternative $d_j \in C_j$ that is $\rho$-affordable for minimum $\rho$ is chosen, breaking ties arbitrarily. Then for each voter $i$ approving $d_j$, the remaining budget $b_i$ of $i$ is set to $\max(0, b_i - \rho)$.

If in some round, no alternative is affordable for any $\rho$, 
MES \emph{terminates prematurely}. Decisions for the remaining rounds can be made arbitrarily (as far as our axioms are concerned), but in practice are done using an appropriate \emph{completion rule} such as Sequential Phragm\'en. In our experiments, we followed the ``$\epsilon$-completion'' strategy introduced by \citet[Sec.~3.4]{peters2021proportional} which approximates the concept of choosing the lowest affordable alternative, retaining the core idea of MES.

It is also possible to run MES in an offline fashion (``Offline MES''), by letting it choose at each step in which round it would like to make a decision (namely, the round where some alternative is $\rho$-affordable for minimal $\rho$, across all alternatives in all rounds where no decision has been taken yet). All our proofs and counterexamples in \Cref{sec:guarantees:mes} below go through for Offline MES, so it also satisfies Weak EJR but fails PJR. This also follows from the results of \citet{masavrik2023group}, who study Offline MES in a more general setting.

\enlargethispage{10pt}
\subsubsection{Proportional Approval Voting (PAV)}
PAV (based on a rule of \citealp{Thie95a}) is an offline rule that selects the decision sequence $D \in C_1 \times \dots \times C_T$ that maximizes
\[
\operatorname{PAV-score}(D) = \smash{\sum_{i\in N}} \:\, 1 + \frac12 + \frac13 + \dots + \frac1{U^i_D},
\] 
where we recall that $U^i_D$ is the number of rounds in which voter $i$ approves the decision of $D$. This harmonic objective function is the unique additive objective that leads to a proportional rule \citep{ejr}.
Finding the optimum decision sequence for PAV is NP-hard \citep[Thm.~5.1]{approvalbasedapportionment}, just like in the multi-winner setting \citep{aziz2015computational}. 

However, one can define a polynomial-time local search variant of PAV, that in the multi-winner setting has been found to satisfy the same proportionality axioms as PAV \citep{aziz2018complexity}.\footnote{One could use a sequential version of PAV to get an online rule, which \citet{page2020electingexecutive} conjectured to satisfy at least Weak JR (a weakening of Weak PJR). However, a counterexample from multiwinner voting \citep[Table 2]{pjr} can be adapted to show that Sequential PAV fails Weak PJR in our model (repeat the profile from their paper for $k = 6$ rounds).
}
The variant starts with an arbitrary decision sequence  and keeps changing the decision in some round if this increases the $\operatorname{PAV-score}$ of the decision sequence by at least $n/T^2$.%
\footnote{This threshold mirrors the $n/k^2$ threshold suggested by \citep{aziz2018complexity} for multi-winner voting, but see \citet[Remark 4.2]{peters2025core} for a situation where a somewhat smaller threshold is necessary.}
The process terminates if there is no possible change that leads to a sufficient increase. 

\begin{definition}[Local-Search PAV]
	Given an instance $(N, R, A, C)$ with $n$ voters and time horizon $T$, a decision sequence $D$ is a possible output of \emph{Local-Search PAV} if there exists some initial decision sequence $D_{\text{init}}$ such that $D$ can be returned by the following algorithm:
	
	\iflatexml
	\begin{algorithm*}[H]
		$D \gets D_{\text{init}}$\;
		\While{there is a round $j \in R$ and an alternative $a_j \in C_j$ such that $\operatorname{PAV-score}\left(D \setminus \{d_j\} \cup\left\{a_j\right\}\right) \geq \operatorname{PAV-score}(D) + \frac{n}{T^2}$   }
		{   
			$D \gets D \setminus \{d_j\} \cup\left\{a_j\right\} $\\
		}
		\Return $D$
	\end{algorithm*}
	\else
	\smallskip
	\begin{tabular}{l|l}
		\hspace{-15pt}
		&
		\hspace{-10pt}
		\makeatletter
		\let\@latex@error\@gobble
		\makeatother
		\begin{algorithm*}[H]
			\vspace{-2pt}$D \gets D_{\text{init}}$\;
			\While{there is a round $j \in R$ and an alternative $a_j \in C_j$ such that $\operatorname{PAV-score}\left(D \setminus \{d_j\} \cup\left\{a_j\right\}\right) \geq \operatorname{PAV-score}(D) + \frac{n}{T^2}$   }
			{   
				$D \gets D \setminus \{d_j\} \cup\left\{a_j\right\} $\\
			}
			\Return $D$
		\end{algorithm*}
	\end{tabular}
	\fi
\end{definition}

\noindent
The threshold of $n/T^2$ is large enough to ensure termination in polynomial time.

\begin{proposition}[\citealp{aziz2018complexity}]
	\label{prop:LSPAVruntime}
	Local-Search PAV terminates in polynomial time.
\end{proposition}
\begin{proof}
	Note that a single improving swap can be found and executed in polynomial time, by iterating through all rounds $j \in R$ and all alternatives $a_j \in C_j$. Now, let us assess how many improvements the local search algorithm may perform. Each improvement increases the total PAV-score of the decision sequence by at least $n/T^2$. The maximum possible PAV-score of a length-$T$ decision sequence is $n \cdot(1+1 / 2+\cdots+1 / T)=O(n \ln T)$. Thus, there can be at most $O( \frac{n \ln T}{n/T^2}  ) = O(T^2 \ln T)$ improving swaps.
\end{proof}

 \section{Satisfying the Axioms}
\label{sec:guarantees}

In this section, we establish which of our three rules satisfy which of our four axioms (see \Cref{tab:results}).
\citet{bulteau2021jr} conjectured that no polynomial-time rule achieves Weak PJR (they gave existence proofs which were based on exponential-time algorithms). Our results show that, in fact, Weak PJR as well as the stronger axioms can be satisfied in polynomial time, by attractive rules, some of which even work online. 

\subsection{Online: Sequential Phragm\'en}
\label{sec:guarantees:phragmen}

We begin by analyzing the sequential Phragm\'en rule. 
\begin{theorem}
    Sequential Phragm\'en satisfies PJR.
\end{theorem}

\begin{proof}
For a contradiction, suppose that $S \subseteq N$ witnesses a violation of PJR: it agrees in rounds $R^* = \{j_1, \dots, j_k\} \subseteq R$ and has size $|S| \ge \ell \cdot \frac{n}{k}$, but there are fewer than $\ell$ rounds in which at least one member of $S$ approves the decision.

First, we claim that if $d_j \in C_j$ is the alternative chosen in some round $j \in R^*$, then $s_{d_j} \leq \frac{\ell }{|S|}$. Note that the total load $\sum_{i \in S} x^i$ assigned to members of $S$ is at most $\ell - 1$ because each decision incurs a load of 1. Since $j\in R^*$, there is an alternative $c' \in C_j$ that everyone in $S$ approves. Hence, 
\[
s_{c'} \leq \frac{1 + \sum_{i \in S}x^i}{|S|} \leq \frac{1 + (\ell-1)}{|S|} = \frac{\ell}{|S|},
\]
where the first inequality follows from the definition of $s_{c'}$ as a minimum. Since Sequential Phragm\'en selected $d_j$, we have $s_{d_j} \le s_{c'}$, showing our claim.

Call a round $j$ a \emph{bad round} if $j \in R^*$ and the decision $d_j$ is not approved by any voter in $S$.
Fix some voter $i \in N \setminus S$, and suppose $i$ gets assigned some load during at least one bad round. Consider the point just after the last bad round $j$ where $i$ is assigned some load. 
At this point we must have $x^i \le \smash{\frac{\ell }{|S|}}$ since otherwise $s_{d_j} > \smash{\frac{\ell }{|S|}}$, contradicting our claim. 
Thus, at most $\smash{\frac{\ell}{|S|}}$ load was assigned to $i$ during all bad rounds together.
Clearly, this last claim is also true for voters $i \in N \setminus S$ who do not get assigned any load during any bad round.

In a bad round, load is only assigned to voters outside $S$ (since the round is bad). Thus, by summing over all $i \in N\setminus S$, we see that the total load assigned in bad rounds is at most
\[
\textstyle
|N \setminus S| \cdot  \frac{\ell}{|S|}  = \frac{|N|}{|S|}\cdot \ell - \frac{|S|}{|S|}\cdot \ell \leq \frac{k}{\ell}\ell - \ell = k - \ell.
\]
However, there are at least $k - (\ell - 1)$ bad rounds, so a total load of at least $k - \ell + 1$ is distributed across bad rounds, a contradiction.
\end{proof}

In multi-winner voting, Sequential Phragm\'en is well-known to fail EJR \citep{brill2023phragmen}, so it is unsurprising that it also fails (Weak) EJR in our setting. While it is possible to adapt the example of \citet{brill2023phragmen} to our setting, we use a somewhat simpler example.\footnote{This example can also be used to show that Sequential Phragm\'en fails EJR for the related settings of approval-based apportionment and of multi-winner voting. We thus add a different EJR counterexample to the literature, which may turn out useful in future work.}

\begin{theorem}
    Sequential Phragm\'en fails Weak EJR.
\label{thm:phragmen_fails_EJR}
\end{theorem}

\begin{proof}
	\iflatexml
	\begin{figure}
		\centering
		\scalebox{0.94}
		{\setlength{\tabcolsep}{4pt}
			\begin{tabular}{lccc}
				\toprule
				Rounds & 1 \ - \ 10 \\
				\midrule
				Voters 1, 2, 3 & $\{\alternativea, \alternativeb\}$ \\
				Voters 4, 5, 6, 7  & $\{\alternativea, \alternativec\}$  \\
				Voters 8, 9  & $\{\alternativeb, \alternativec\}$ \\
				Voter 10 & $\{\alternativeb\}$  \\
				\bottomrule
		\end{tabular}}
		\caption{Instance where Sequential Phragm\'en fails Weak EJR. The approval sets are the same in all rounds.}
		\label{fig:phragmen_fails_EJR_simpler_example}
	\end{figure}
	\else
	\begin{wrapstuff}[r,type=figure,width=6.3cm]
		\centering
		\scalebox{0.94}
		{\setlength{\tabcolsep}{4pt}
			\begin{tabular}{lccc}
				\toprule
				Rounds & 1 \ - \ 10 \\
				\midrule
				Voters 1, 2, 3 & $\{\alternativea, \alternativeb\}$ \\
				Voters 4, 5, 6, 7  & $\{\alternativea, \alternativec\}$  \\
				Voters 8, 9  & $\{\alternativeb, \alternativec\}$ \\
				Voter 10 & $\{\alternativeb\}$  \\
				\bottomrule
		\end{tabular}}
		\caption{Instance where Sequential Phragm\'en fails Weak EJR. The approval sets are the same in all rounds.}
		\label{fig:phragmen_fails_EJR_simpler_example}
	\end{wrapstuff}
	\fi
	Consider the instance shown in \Cref{fig:phragmen_fails_EJR_simpler_example}, with $T = 10$ rounds and $n = 10$ voters. Here, the same profile is repeated for all the rounds, so it is also an example where Sequential Phragm\'en fails EJR in the setting of approval-based apportionment \citep{approvalbasedapportionment}. In this example, Sequential Phragm\'en alternates between two alternatives as winner in each round and produces the decision sequence $D = (\alternativea, \alternativeb, \alternativea,\alternativeb, \alternativea, \alternativeb, \alternativea, \alternativeb, \alternativea, \alternativeb)$. The first three voters approve all the decisions of the outcome, so each of them gets a satisfaction of $10$ while the rest of the voters approve only half of the decisions from $D$, so they each get a satisfaction of $5$.  The coalition $S = \{4, 5, 6, 7, 8, 9\}$ agrees in all rounds on $\alternativec$, so Weak EJR demands that at least one voter in $S$ should approve at least $\ell = 6$ decisions (since $6 = |S| \ge \ell \cdot \frac{n}{T} = 6 \cdot \frac{10}{10}$). However, each voter in $S$ approves the decision of only $5$ rounds, which is strictly less than $\ell = 6$, so Weak EJR is violated. Hence, Sequential Phragm\'en fails Weak EJR.
\end{proof}

\subsection{Semi-online: Method of Equal Shares (MES)}
\label{sec:guarantees:mes}

We do not know an online rule that satisfies Weak EJR. In fact, in \Cref{app:online-ejr}, we give a reduction showing that such a rule could be converted to a multi-winner voting rule satisfying Weak EJR and the axiom of \emph{committee monotonicity}, the existence of which is a major open problem \citep{lacknerskowron,sanchez2017monotonicity}.%
\footnote{\citet{elkind2024verifying} recently proved that no (semi-)online rule can satisfy EJR, answering a question we asked in earlier versions of the paper.}

Can we evade this difficulty with some foresight, by relaxing the online requirement?  Knowing the time horizon $T$ is a common assumption found in online learning settings like multi-armed bandits \cite{barman2023fairness}. Indeed, if we know the total number of rounds $T$, we can use MES (which is online except that it uses $T$ to determine the price $p = n/T$ of deciding a round). We show that it satisfies Weak EJR.

\begin{theorem}
    MES satisfies Weak EJR.
\end{theorem}
\begin{proof}
	Write $p = \frac{n}{T}$ for the cost of deciding a round according to MES. Suppose $S \subseteq N$ witnesses a violation of Weak EJR, where $|S| \ge \ell \cdot p$ but every $i \in S$ approves the decision of at most $\ell - 1$ rounds.
	
	Suppose MES stops prematurely without making a decision in round $j \in R$. At that point, some voter $i \in S$ has a remaining budget of strictly less than $p/|S|$, as otherwise $S$ has a combined remaining budget of at least $p$ and could therefore purchase an alternative that $S$ agrees on in round $j$. Otherwise, if it does not stop prematurely, MES makes a decision for all $T$ rounds, in which case all available money is spent (since $T \cdot p = n$), and hence we can again find $i \in S$ with remaining budget 0, i.e., less than $p/|S|$.
	
	Thus, $i$ has spent more than $1 - p/|S|$ by the time that MES has terminated, and has used that money to pay in at most $\ell-1$ rounds. In those rounds, $i$ therefore paid on average strictly more than 
	\[
	\smash{\frac{1 - p/|S|}{\ell-1} \geq \frac{1 - 1/\ell}{\ell-1}=\frac{1}{\ell}}.
	\] 
	Hence, there must be a round when $i$ paid strictly more than $1/\ell$ for the decision; let $j \in R$ be the first such round with decision $d_j \in C_j$. Just before paying for $d_j$, every voter in $S$ had at least $1/\ell$ budget left as they each so far paid at most $1/\ell$ for at most $\ell-1$ alternatives. 
	But note that the alternative $c \in C_j$ on which $S$ agrees is therefore $1/\ell$-affordable because 
	\[
	\smash{|S| \cdot \frac{1}{\ell} \geq |S|\cdot \frac{n}{|S|T} = p},
	\]
	while $d_j$ is not $1/\ell$-affordable (since $i$ had to pay strictly more than $1/\ell$ for $d_j$). This is a contradiction to MES choosing the alternative that is $\rho$-affordable for the lowest $\rho$.
\end{proof}

Does MES provide good guarantees for coalitions that do not agree on all rounds? Unfortunately not. We show that MES fails PJR and EJR. This is perhaps surprising since, in other settings, MES usually satisfies at least as many proportionality axioms as Sequential Phragm\'en. The reason for its failure here is that coalitions may agree only in early rounds where MES greedily maximizes efficiency, and then MES does not satisfy the fairness requirements in subsequent rounds where there may not be enough agreement between voters to support the purchase of any alternative. 

We begin by showing that MES does not guarantee EJR.

\begin{theorem}
	MES fails EJR. It does so even on instances where it does not terminate prematurely, i.e., where it outputs a decision for all rounds.
\end{theorem}

\begin{proof}
	\iflatexml
\begin{figure}
	\centering
	{\setlength{\tabcolsep}{4pt}
		\begin{tabular}{lccccccc}
			\toprule
			Rounds & 1 - 10 & 11 - 16 \\
			\midrule
			Voter 1 & $\{\alternativea\}$ &  $\{\alternativec\}$ \\
			Voter 2 & $\{\alternativea\}$ &  $\{\alternatived\}$ \\
			Voter 3 & $\{\alternativea\}$ &  $\{\alternativee\}$ \\
			Voters 4, 5, 6, 7, 8 & $\{\alternativeb\}$ &  $\{\alternativeb\}$ \\
			\bottomrule
	\end{tabular}}
	\caption{Instance where MES does not terminate prematurely yet fails EJR.}
	\label{fig:MES_fails_StrongEJR}
\end{figure}
	\else
\begin{wrapstuff}[r,type=figure,width=7cm]
	\centering
	{\setlength{\tabcolsep}{4pt}
		\begin{tabular}{lccccccc}
			\toprule
			Rounds & 1 - 10 & 11 - 16 \\
			\midrule
			Voter 1 & $\{\alternativea\}$ &  $\{\alternativec\}$ \\
			Voter 2 & $\{\alternativea\}$ &  $\{\alternatived\}$ \\
			Voter 3 & $\{\alternativea\}$ &  $\{\alternativee\}$ \\
			Voters 4, 5, 6, 7, 8 & $\{\alternativeb\}$ &  $\{\alternativeb\}$ \\
			\bottomrule
	\end{tabular}}
	\vspace{-2pt}
	\caption{Instance where MES does not terminate prematurely yet fails EJR.}
	\label{fig:MES_fails_StrongEJR}
\end{wrapstuff}
\fi
	Consider the instance shown in \Cref{fig:MES_fails_StrongEJR}. We have $T = 16$ rounds and $n = 8$ voters. The group of voters $S_1 = \{1, 2, 3\}$ agrees in the first $10$ rounds (jointly approving $\alternativea$) while the group $S_2 = \{4, 5, 6, 7, 8\}$ agrees in all rounds (approving $\alternativeb$). The budget of each voter $i$ is $b^i = 1$ unit while the price of each round is $p = \frac{n}{T} = \frac{8}{16} = 0.5$ units. 
	
	In round $1$, $\alternativea$ is $\rho$-affordable for $\rho = \frac16$ whereas $\alternativeb$ is $\rho$-affordable for $\rho = \frac{1}{10}$. Thus, $\alternativeb$ is selected in round $1$. The remaining budget of voters in $S_2$ is each set to $b^i - \rho = 1 - \frac{1}{10} = \frac{9}{10}$. By the same calculation, it can be seen that $\alternativeb$ is selected in every round until round $10$, after which every voter in $S_2$ has exhausted their budget. 
	
	Beginning in round $11$, only voters in $S_1$ have budget remaining. However, the voters $1, 2$ and $3$ do not agree in these remaining rounds. Further, every alternative apart from $\alternativeb$ is approved by only one voter, so its supporter will have to pay the full cost of $p = 0.5$ if it is selected. Thus, across the rounds $11$-$16$, each of the $3$ voters in $S_1$ will get to select the outcome of $2$ of the rounds (with the exact sequence depending on how ties are broken). 
	
	Thus, the final decision sequence $D$ will have $\alternativeb$ for the first 10 rounds and $\alternativec, \alternatived, \alternativee$ each appearing twice in the last $6$ rounds (in some order). Now notice that $S_1$ agrees in the first $k = 10$ rounds, so with $\ell = 3$, since $|S_1| \ge \ell \cdot \frac{n}{k} = 3 \cdot \frac{8}{10}$, EJR requires that some voter in $S_1$ approves decisions of at least $\ell = 3$ rounds. However, each voter in $S_1$ approves exactly $2$ decisions in $D$, violating EJR.
\end{proof}

We also have a similar example where MES fails PJR. However, we note that this example (shown in \Cref{fig:MESfails}) depends on empty approval sets and causes MES to terminate prematurely. We discuss these two caveats further below.

\begin{theorem}
	MES fails PJR.
	\label{thm:MESfails}
\end{theorem}

\begin{proof}
\iflatexml
\begin{figure}
	\centering
	\scalebox{0.9}
	{\setlength{\tabcolsep}{4pt}
		\begin{tabular}{lcccccc}
			\toprule
			Round & 1 \& 2 \& 3 & 4 \& 5 \& 6 \\
			\midrule
			Voter 1 & $\{\alternativea\}$ & $\emptyset$ \\
			Voter 2 & $\{\alternativeb\}$ &  $\{\alternativeb\}$ \\
			Voter 3 & $\{\alternativeb\}$ &  $\{\alternativeb\}$ \\
			\bottomrule
	\end{tabular}}
	\caption{Instance where MES fails PJR (and JR). Note the empty approval set of voter $1$ in rounds $4, 5, 6$.}
	\label{fig:MESfails}
\end{figure}
\else
\begin{wrapstuff}[r,type=figure,width=6.5cm]
	\centering
	\scalebox{0.9}
	{\setlength{\tabcolsep}{4pt}
		\begin{tabular}{lcccccc}
			\toprule
			Round & 1 \& 2 \& 3 & 4 \& 5 \& 6 \\
			\midrule
			Voter 1 & $\{\alternativea\}$ & $\emptyset$ \\
			Voter 2 & $\{\alternativeb\}$ &  $\{\alternativeb\}$ \\
			Voter 3 & $\{\alternativeb\}$ &  $\{\alternativeb\}$ \\
			\bottomrule
	\end{tabular}}
	\caption{Instance where MES fails PJR (and JR). Note the empty approval set of voter $1$ in rounds $4, 5, 6$.}
	\label{fig:MESfails}
\end{wrapstuff}
\fi
Consider the instance shown in \Cref{fig:MESfails}, with $T = 6$ rounds and $n = 3$ voters. Here, voter $1$ approves $\alternativea$ in the first $3$ rounds but does not approve any alternative in the remaining rounds. Voters $2$ and $3$ approve $\alternativeb$ in all rounds. 
The budget of each voter $i$ is $b^i = 1$ unit while the price of each round is $p = \frac{n}{T} = \frac36 = 0.5$ units. 

In round $1$, $\alternativea$ is $\rho$-affordable for $\rho = 0.5$ as only voter 1 approves $\alternativea$ and thus has to bear its whole cost of $p = 0.5$. However, $\alternativeb$ is $\rho$-affordable for $\rho = 0.25$ as both voters $2$ and $3$ approve it, so they can equally divide the cost $p$. So $\alternativeb$ is selected, and the remaining budget of voters $2$ and $3$ is set to $0.75$. By the same computation, $\alternativeb$ is selected in rounds $2$ and $3$, after which voters $2$ and $3$ have a remaining budget of $0.25$. In round $4$, the only alternative to elect is $\alternativeb$ and voters $2$ and $3$ have just enough budget to buy it at $\rho = 0.25$, leading to exhaustion of their whole budget. In round $5$, voters $2$ and $3$ do not have any budget to buy any alternative while voter $1$ has the budget but does not approve any alternative. Thus, MES terminates prematurely. Any completion rule used from this point can only elect $\alternativeb$ in the remaining rounds as that is the only available alternative. Thus, we obtain the decision sequence $D = (\alternativeb, \alternativeb, \alternativeb, \alternativeb, \alternativeb, \alternativeb)$.

Now, the first voter forms an individual group $S = \{1\}$ which agrees in the first 3 rounds (as voter $1$ approves $\alternativea$). With $\ell = 1$, it has size $|S| = 1 \ge \ell \cdot \frac{n}{3} = 1 \cdot \frac33$. Thus, PJR requires that in at least $\ell = 1$ rounds, the outcome is $\alternativea$. However, MES (with any completion rule) never elects $\alternativea$, so MES fails PJR.
\end{proof}

\begin{remark}
	The counterexample of \Cref{thm:MESfails} used a requirement of $\ell = 1$ rounds, so it also shows that MES fails JR.
\end{remark}

This also implies that MES fails EJR. \citet{elkind2024verifying} show that in fact every semi-online method fails EJR.

As we mentioned, in the example of \Cref{fig:MESfails}, MES terminates prematurely. Thus, a completion rule would need to be used to make decisions in the remaining rounds. Due to the empty approval sets of Voter 1 in rounds 4, 5, and 6, however, it is clear that no matter how the remaining rounds are decided, the resulting decision sequence will fail PJR. 

\begin{figure}
	\centering
	\scalebox{1}
	{\setlength{\tabcolsep}{4pt}
		\begin{tabular}{lcc}
			\toprule
			Round & 1--7 & 8--10 \\
			\midrule
			Voters 1--70 & $\{\alternativea\}$ &  $\{\alternativec\}$  \\
			Voters 71--100 & $\{\alternativeb\}$ &  all disjoint \\
			\bottomrule
	\end{tabular}}
	\caption{Example illustrating how MES fails PJR (and JR) without empty approval sets. The cell ``all disjoint'' indicates that each of the last 30 voters approves a unique alternative not approved by anyone else in the last 3 rounds. On this example, MES selects $\protect\alternativea$ in the first 7 rounds and then terminates prematurely. Most completion methods select $\protect\alternativec$ in the remaining rounds (including Phragm\'en completion, utilitarian completion, and Add1).}
	\label{fig:MESfailsNonEmpty}
\end{figure}

What about if we insist that all approval sets must be non-empty? In that case we are able to construct examples similar to \Cref{fig:MESfails} where MES terminates prematurely (see \Cref{fig:MESfailsNonEmpty}), and all the standard completion methods lead to a decision sequence that fails PJR. However, we are not aware of an example where PJR is failed for \emph{all possible} completions of the MES output.

It turns out that the premature termination is, in a sense, the cause for the PJR failure: We show that if MES does not terminate prematurely (i.e., it outputs decisions for all the $T$ rounds), its output decision sequence $D$ satisfies PJR.           

\begin{theorem}
	If MES outputs a decision sequence $D$ which contains a decision for all rounds (without terminating prematurely), then for all $S \subseteq N$ and $\ell \in \mathbb{N}$ with $|S| \ge \ell \cdot \frac{n}{T}$, there are at least $\ell$ rounds in which at least one member of $S$ approves the decision of $D$ (no matter in how many rounds $S$ agrees). In particular, $D$ satisfies PJR.
\end{theorem}

\begin{proof}
	Let $\ell \in \mathbb{N}$ and let $S \subseteq N$ be a group of voters with $|S| \ge \ell \cdot \frac{n}{T}$.
	Recall that MES initially gives each voter a budget of $1$ while each round's decision has a price of $p = \frac{n}{T}$. As MES produced a decision for all rounds, in total $p \cdot T = n$ amount was spent. This implies that every voter's budget of $1$ is used in full. Thus, the voters in $S$ paid a total amount of $|S|$ during the execution of MES. Recall that voters only pay when they approve a round's decision. Thus, the number of rounds in which the decision is approved by at least 1 voter in $S$ must be at least $|S|/p = |S| \cdot \frac{T}{n} \ge \ell$ rounds, as desired.
\end{proof}

\subsection{Offline: Proportional Approval Voting (PAV)}
\label{sec:guarantees:pav}

In some settings, offline voting (where alternatives and approval sets for all rounds are known in advance) is possible, e.g., for voting in combinatorial domains with independent issues.
Studying the offline setting can also clarify which axioms are plausible aims for online rules.
It turns out that if we make decisions fully offline, there is a rule that satisfies all four of our axioms: PAV,
as well as its polynomial-time local search variant. The proof uses a swapping argument: if the output violates EJR, then in at least 1 round, one can change the decision and thereby increase the PAV objective function. This technique is also used in multi-winner voting \citep{ejr}.
Our theorem was recently generalized by 
\citet{masavrik2023group} to show that PAV satisfies an EJR notion in a general model of voting with matroid constraints (noting that our model can be described as involving constraints that in each round, just a single alternative can be chosen).

\begin{theorem}
	\label{thm:pav-satisfies-strong-ejr}
	PAV and Local-Search PAV satisfy EJR.
\end{theorem}

\begin{proof}
	If $D_{\text{PAV}}$ is the decision outcome chosen by PAV, and we run Local-Search PAV with $D_{\text{init}} = D_{\text{PAV}}$, then it immediately terminates. Hence $D_{\text{PAV}}$ is a possible output of Local-Search PAV. Thus, to prove the theorem, it is sufficient to prove that every possible output of Local-Search PAV satisfies EJR.
	
	So let $D$ be a decision outcome returned by Local-Search PAV with $|D| = T$. 
	For a contradiction, suppose that $S \subseteq N$ witnesses a violation of EJR of $D$, with $S$ agreeing on a set $R^* \subseteq R$ of $|R^*| = k$ rounds and with $|S| = s \ge \ell \cdot \frac{n}{k}$ but $U^i_D < \ell$ for all $i \in S$. 
	For any round $r \in R^*$, let $a_r \in C_r$ be an alternative on which $S$ agrees and $d_r$ be the current decision for that round.
	If we replace $d_r$ by $a_r$ in $D$, the change in the $\operatorname{PAV-score}$ of $D$ is
	\begin{align}
		\Delta(a_r, d_r) 
		&= \underbrace{ \sum_{\substack{i \in S \\ \ d_r \notin A^i_r}} \frac{1}{U^i_D + 1} + \sum_{\substack{i \in N \setminus S \\ \ d_r \notin A^i_r \\ a_r \in A^i_r}} \frac{1}{U^i_D + 1} }_{\text{increase due to addition of } a_r} 
		\quad -
		\underbrace{\quad \sum_{\substack{i \in N \setminus S \\ \ d_r \in A^i_r \\ a_r \notin A^i_r}} \frac{1}{U^i_D} \quad}_{ \text{due to removal of } d_r} 
		\nonumber
		\\
		& \ge \sum_{\substack{i \in S \\ \ d_r \notin A^i_r}} \frac{1}{U^i_D + 1} - \sum_{\substack{i \in N \setminus S \\ \ d_r \in A^i_r}} \frac{1}{U^i_D}
		\label{eq:delta_increase_PAV}
	\end{align}
	
	Note that the change for a particular round may be negative. Summing this change over $R^*$, the set of rounds where $S$ agrees, we get 
	\begin{align*}
		\sum_{r \in R^*}\Delta(a_r, d_r) &\ge \sum_{r \in R^*} \Bigg( \sum_{\substack{i \in S \\ \ d_r \notin A^i_r}} \frac{1}{U^i_D + 1} \ - \quad \sum_{\substack{i \in N \setminus S \\ \ d_r \in A^i_r}} \frac{1}{U^i_D} \Bigg ) \tag{from \eqref{eq:delta_increase_PAV}} \\ 
		&= \sum_{i \in S} \sum_{\substack{r \in R^* \\ \ d_r \notin A^i_r}} \frac{1}{U^i_D + 1} \ - \ \sum_{i \in N \setminus S} \sum_{\substack{r \in R^* \\ \ d_r \in A^i_r}} \frac{1}{U^i_D} \tag{interchanging  sums}  
	\end{align*}
	
	Write $U^i_{R^*}$ for the number of rounds in $R^*$ in which $i$ approves the decision of $D$.
	Note that for all $i \in S$, we have $U^i_{R^*} \le U^i_D \le \ell - 1$, and so the number of decisions in rounds $R^*$ not approved by $i$ is at least $k - (\ell - 1) = k - \ell + 1$. Hence, we get
	
	\begin{align*}
		\sum_{r \in R^*}\Delta(a_r, d_r) 
		&\ge  \sum_{i \in S} \frac{k - \ell + 1}{U^i_D + 1} \ - \sum_{i \in N \setminus S} \sum_{\substack{r \in R^*\\ \ d_r \in A^i_r}} \frac{1}{U^i_D}  \\
		&\ge  \sum_{i \in S} \frac{k - \ell + 1}{U^i_D + 1} - \sum_{i \in N \setminus S} \  \sum_{\substack{r \in R^*\\ \ d_r \in A^i_r}}  \frac{1}{U^i_{R^*}} 
		\tag{since $U^i_{R^*} \le U^i_D$ for all $i \in N$} \\
		&= \sum_{i \in S} \frac{k - \ell + 1}{U^i_D + 1} \  - \sum_{i \in N\setminus S} 1 
		\tag{definition of $U^i_{R^*}$} \\
		&\ge \sum_{i \in S} \frac{k - \ell + 1}{\ell - 1 + 1} - (n - s) 
		\tag{since $U^i_D \le \ell - 1$ for all $i \in S$} \\ 
		&=  \bigg( \frac{k - \ell + 1}{\ell} \bigg) \cdot s -  (n - s)  \\
		&= \frac{sk}{\ell} - s + \frac{s}{\ell} - n + s \\
		&= s\frac{k}{\ell} + \frac{s}{\ell} - n \ge n + \frac{s}{\ell} - n \ge \frac{n}{k} \tag{since $s \ge \ell \cdot \frac{n}{k}$}
	\end{align*}
	
	Thus, for an average $r \in R^*$, we have $\Delta(a_r, d_r) \ge \frac{1}{k} \cdot \frac{n}{k} = \frac{n}{k^2} \ge \frac{n}{T^2}$ using that $|R^*| = k \le T$. Therefore, there exists a round $r \in R^*$ such that $\Delta(a_r, d_r) \ge \frac{n}{T^2}$. Hence $D$ admits a swap that increases its PAV-score by at least $\frac{n}{T^2}$, and thus Local-Search PAV would not terminate and return $D$. This gives a contradiction and completes the proof.
\end{proof}

 \section{Impossibility of Stronger Guarantees}
\label{sec:hardness}

One could try to strengthen the axioms we have discussed in various ways, by attempting to either increase the amount of utility guarantees, or to relax the constraints on group size. However, one can show that many such strengthenings lead to axioms that may be infeasible. 
For example, if we reduce the requirement on the coalition size by $\epsilon$ in the definition of PJR, there are instances where no decision sequence satisfies the strengthening. 

\begin{restatable}{theorem}{Stronger PJR need not exist}
    \label{thm:counterexample}
	Let $\epsilon > 0$. 
	Then there exists an instance where no decision sequence $D$ satisfies ``\emph{$\epsilon$-Stronger PJR}'', defined to require that 
	for every $\ell \in \mathbb N$ and every group of voters $S \subseteq N$ that agrees in $k$ rounds and has size $|S| \geq (\ell  - \epsilon)\cdot \frac{n}{k}$, there are at least $\ell$ rounds $j \in R$ in which the decision $d_j$ of $D$ is approved by at least one voter in $S$ (i.e., $d_j \in \bigcup_{i \in S} \smash{A_j^i}$).
\end{restatable}

The proof of \Cref{thm:counterexample} is in \Cref{app:infeasible-axioms}, where we also discuss several other potential strengthenings.
The proof constructs a counterexample where in the first $k$
rounds, there are very many coalitions that agree (but cannot all be satisfied simultaneously),
while in the remaining rounds, there is no agreement at all (which makes it impossible to satisfy
all the justified demands from earlier rounds). 

\Cref{thm:counterexample} is a worst case result, and we could hope for
rules that satisfy $\epsilon$-Stronger PJR on inputs where it is possible. However, we
show that no (semi-)online rule can do this.

\begin{restatable}{theorem}{Stronger PJR cannot be satisfied online}
    \label{thm:counterexample-online}
	Let $\epsilon > 0$. No semi-online decision rule returns an \emph{$\epsilon$-Stronger PJR} decision sequence whenever one exists. 
\end{restatable}

Again, the proof can be found in \Cref{app:infeasible-axioms}.

\section{Experiments}
\label{sec:experiments}

To understand the performance of our methods empirically, we run our methods on both synthetic and real-world datasets. 

\paragraph{Rules.}
In addition to our proposed rules (Phragm\'en, MES, PAV), we also consider two rules proposed by \citet{lackner2020perpetual}: \emph{Perpetual Quota} (aims at granting each voter a satisfaction as close as possible to their ``quota'') and \emph{Perpetual Consensus} (similar to  Sequential Phragm\'en but strictly enforces an equal distribution of the load incurred).\footnote{These rules intuitively aim for proportionality, but do not satisfy any of our axioms as they even fail the weaker axioms proposed by \citet{lacknermaly2023}.}
We chose those rules as they performed well in the experiments of \citet{lackner2020perpetual}.
Further, we consider two baselines: \emph{Approval Voting} (chooses the alternative with the highest approval score in each round) and \emph{Round Robin} (in each round $j$, voter 
$j \text{ mod } T$
chooses an approved alternative (the first one according to a fixed tie-breaking order)). 

We implemented all the voting rules in Python 3.8, building on existing implementations: For Sequential Phragm\'en  (a.k.a. Perpetual Phragm\'en), Perpetual Quota, Perpetual Consensus, Approval Voting, and Round Robin (a.k.a. Serial Dictatorship), we used the \href{https://github.com/martinlackner/perpetual}{\texttt{perpetual}} python package.
For PAV, we used a standard integer linear-program (ILP) encoding \citep[see, e.g.,][]{peters2020preferences}, solved using Gurobi 9. For MES, our implementation builds upon the \href{https://github.com/Grzesiek3713/pabutools}{\texttt{Grzesiek3713/pabutools}} package and uses Phragm\'en completion in case of premature termination.

\paragraph{Metrics.}
We evaluate our rules on several metrics of voter utility to complement our theoretical guarantees. For comparability of results, we normalize utility and define a single voter $i$'s utility as the \emph{fraction} of rounds in which the voter approves the decision: $\smash{U^i_D}/T$. Based on this, we report three metrics:
\begin{itemize}[leftmargin=15pt]
	\item \emph{Average Utility} of the voters (equivalent to utilitarian social welfare).
\item \emph{Utility of the 25th Percentile}: We sort the vector of utilities and report its 25th percentile. This  is inspired by egalitarian social welfare, which we did not use in our experiments because the minimum utility was often zero. \item \emph{Gini Coefficient}: This metric quantifies the level of inequality in the voter utilities. A lower value corresponds to a more equal utility distribution (with 0 being obtained in case every voter has the same utility). Formally, the Gini coefficient of a decision sequence $D$ is 
\[
%    \textstyle
    \text{gini}(D)= \frac{1}{\sum_{i \in N} U^i_D} \sum_{i \in N} \sum_{j \in N} \frac{1}{2n} \big| U^i_D - D^j_D \big|.
    \]

\end{itemize}

\subsection{Synthetic Data}
\label{sec:experiments:synthetic}

For analysis based on synthetic data, we follow a similar setup to the one used by \citet{lackner2020perpetual} which is based on the popular approach of sampling both voters and alternatives as points in a two-dimensional Euclidean space \citep{elkind2017multiwinner}. 
Because we are particularly interested in group representation, we use distributions of voter locations that have several clusters. 

More precisely, both voters and alternatives are sampled as points in $\mathbb{R}^2$. In each round, a fresh set of alternatives is sampled uniformly in the square $[-1, -1] \times [1, 1]$. The voter locations stay fixed across rounds, and for each voter $i \in N$, we sample $i$'s location independently of other voters from a bivariate normal distribution $\mathcal{N}(x^*_i, y^*_i)$ centered at a point $(x^*_i, y^*_i)$ with standard deviation $\sigma = 0.2$ (unless stated otherwise). We present results for four distributions of voter locations, which differ in their choices of points $(x^*_i, y^*_i)$ and the size of the groups:

\begin{itemize}
	\item \emph{Restricted}: The voters are split into $2$ disjoint groups $S_1$ and $S_2$ with $\frac13$ of voters in $S_1$ and $\frac23$ of voters in $S_2$. The locations of voters in $S_1$ are sampled from $\mathcal{N}(-0.5, -0.5)$, while the locations of voters in $S_2$ are sampled from $\mathcal{N}(0.5, 0.5)$. We further \emph{restrict} the voter locations to lie in the square $[-1, -1] \times [1, 1]$, if necessary resampling from the location distributions until this constraint is satisfied.
	
	\item \emph{Many Groups}: The voters are split into $4$ disjoint groups concentrated in different areas. The first $3$ groups each form $20\%$ of the population while the last group forms $40\%$. For each group, both $x$- and $y$-coordinates are drawn independently from $\mathcal{N}(\pm 0.5, \pm 0.5)$. 
	
	\item \emph{Unbalanced}:  The voters are split into $2$ disjoint groups $S_1$ and $S_2$ with $20\%$ of voters in $S_1$ and $80\%$ of voters in $S_2$. The locations of voters in $S_1$ are sampled from $\mathcal{N}(-0.5, -0.5)$, and the locations of voters in $S_2$ are sampled from $\mathcal{N}(0.5, 0.5)$. For this distribution, we used standard deviation $\sigma = 0.1$.
	
	\item \emph{Balanced and Nearby}: The voters are split into $2$ disjoint groups $S_1$ and $S_2$ with $60\%$ of voters in $S_1$ and $40\%$ of voters in $S_2$. In this distribution, the groups are quite close to each other in comparison to other distributions. All the voters of $S_1$ were sampled from $\mathcal{N}(-0.25, 0)$ while the voters of $S_2$ were sampled from $\mathcal{N}(0.25, 0)$.
\end{itemize}

These four distributions are visualized in the blue pictures in the right-most column of \Cref{fig:param_setting_2}, where we sampled 2\,500 points from these distributions (chosen for visualization purposes; the experiments use smaller numbers of voters).

Our experiments are parameterized by a factor $f \ge 1$, which determines how many alternatives are approved by voters. Specifically, voters approve all alternatives whose Euclidean distance is within $f$ times the distance to their closest alternative. 
For each of the distributions mentioned above, we consider the following three parameter settings:
\begin{itemize}
	\item $n = 20$ voters, $T = 50$ rounds and each round having $20$ alternatives, with $f = 1.5$. This led to an average approval set size of $2.12$. 
	\item $n = 50$ voters, $T = 50$ rounds and each round having $40$ alternatives, with $f = 1.5$. This led to an average approval set size of $2.2$. 
	\item $n = 50$ voters, $T = 50$ rounds and each round having $40$ alternatives, with $f = 3$. This led to a larger approval set size with the average being $7.52$. 
\end{itemize}

We repeated each experiment for $1\,000$ trials. We report our results for the various distributions in \Cref{fig:param_setting_2,fig:param_setting_3,fig:param_setting_4}. Each figure corresponds to a specific distribution and contains multiple subfigures to show the plots for different parameter values. In each subfigure, bar plots are shown for each rule and metric. The length of the bar represents the median value across all the $1000$ trials for the rule while the error bars represent the $25$th and $75$th percentile, with the numeric text after the bar showing the mean value. 

\begin{figure}
	\centering
	
	\begin{subfigure}{\textwidth}
		\centering
		\includegraphics[height=3.6cm]{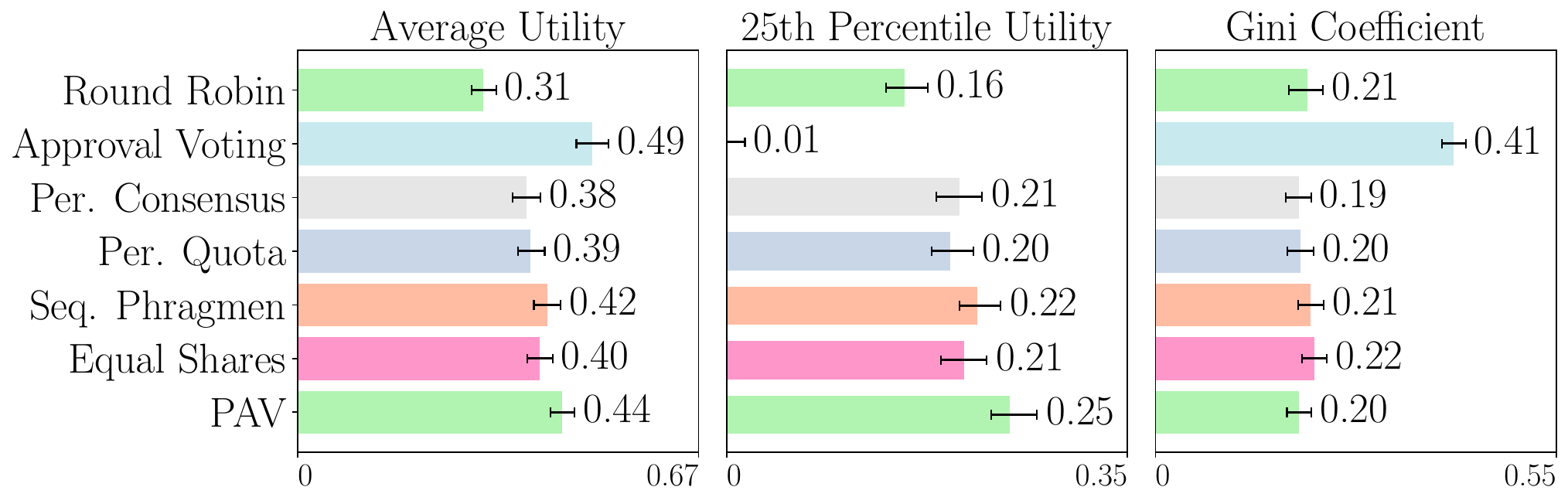}
		\:\raisebox{7pt}{\includegraphics[height=3.2cm]{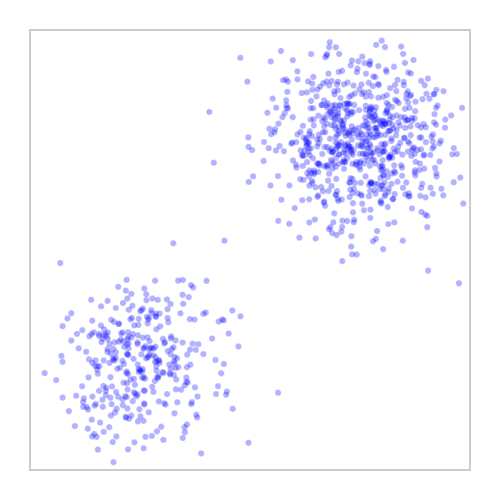}}
		\caption{\emph{Restricted} distribution.}
		\label{fig:euclidean_1000,20,20,50,0.2,eucl1,uniform_square,1.5}
	\end{subfigure}
	
	\vspace{\baselineskip}
	
	\begin{subfigure}{\textwidth}
		\centering
		\includegraphics[height=3.6cm]{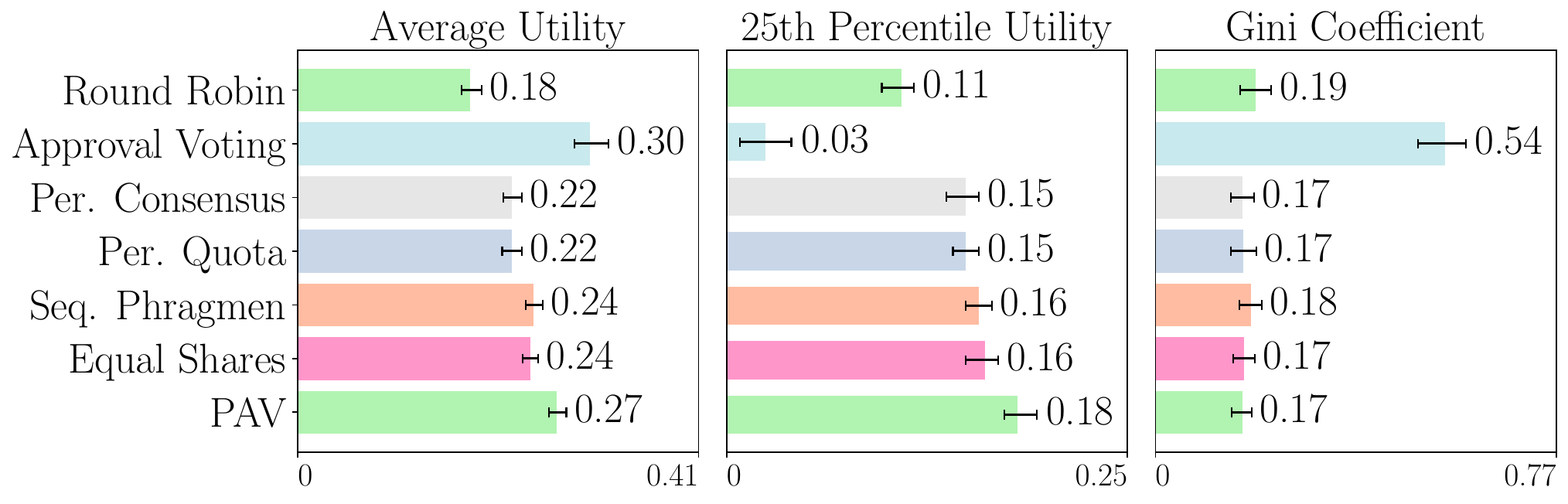}
		\:\raisebox{7pt}{\includegraphics[height=3.2cm]{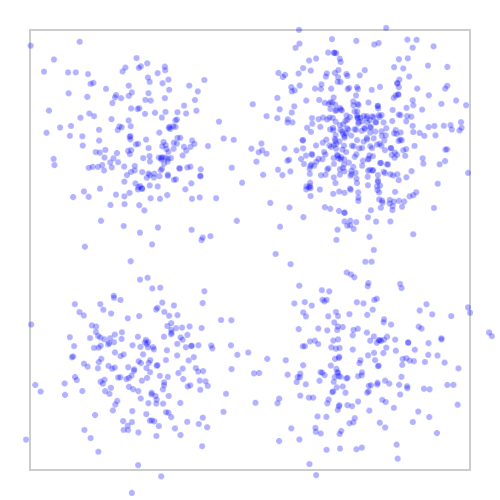}}
		\caption{\emph{Many Groups} distribution.}
		\label{fig:euclidean_1000,20,20,50,0.2,eucl3,uniform_square,1.5}
	\end{subfigure}
	
	\vspace{\baselineskip}
	
	\begin{subfigure}{\textwidth}
		\centering
		\includegraphics[height=3.6cm]{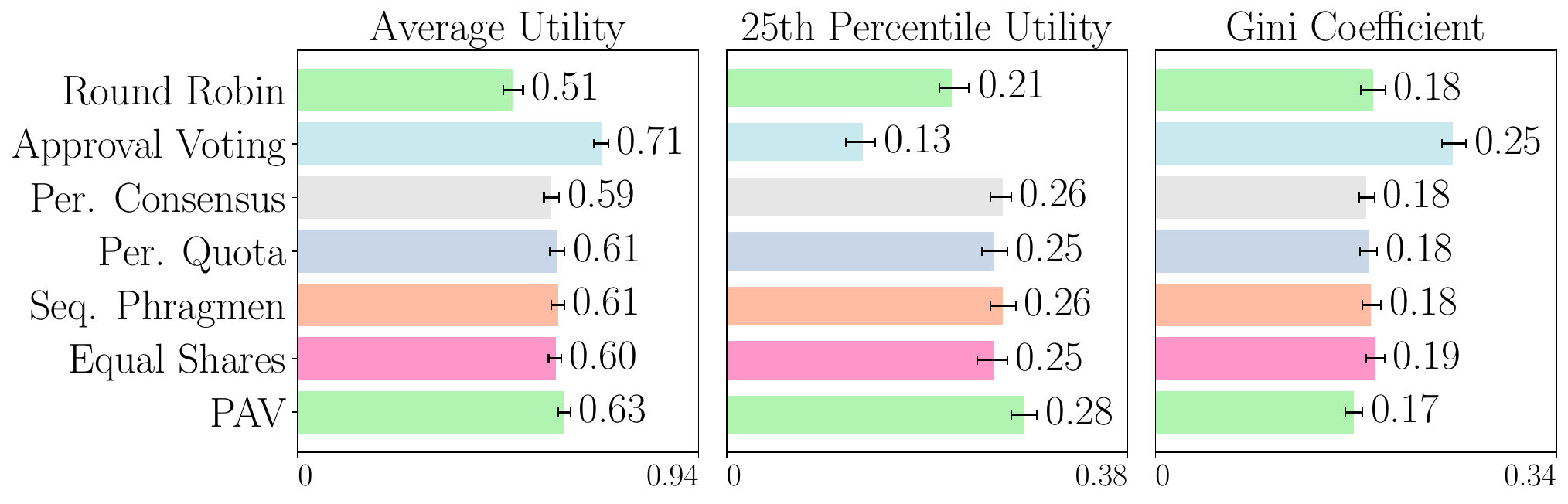}
		\:\raisebox{7pt}{\includegraphics[height=3.2cm]{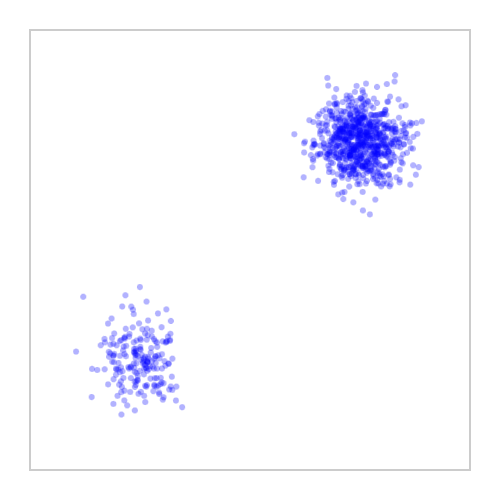}}
		\caption{\emph{Unbalanced} distribution.}
		\label{fig:euclidean_1000,20,20,50,0.1,eucl4,uniform_square,1.5}
	\end{subfigure}
	
	\vspace{\baselineskip}
	
	\begin{subfigure}{\textwidth}
		\centering
		\includegraphics[height=3.6cm]{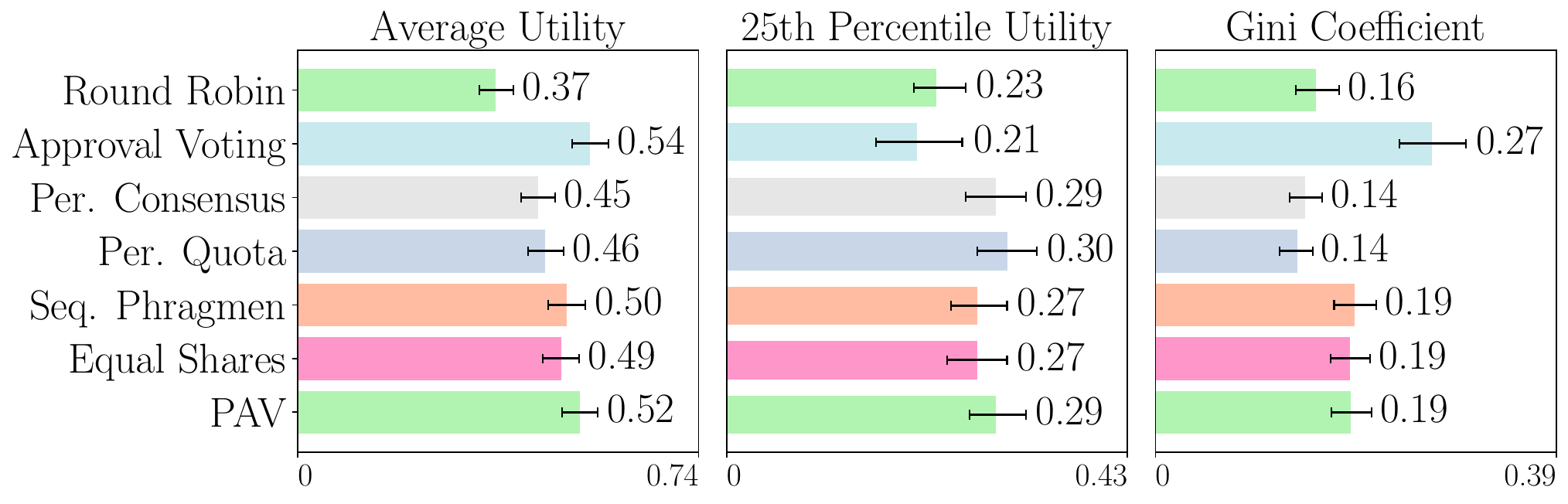}
		\:\raisebox{7pt}{\includegraphics[height=3.2cm]{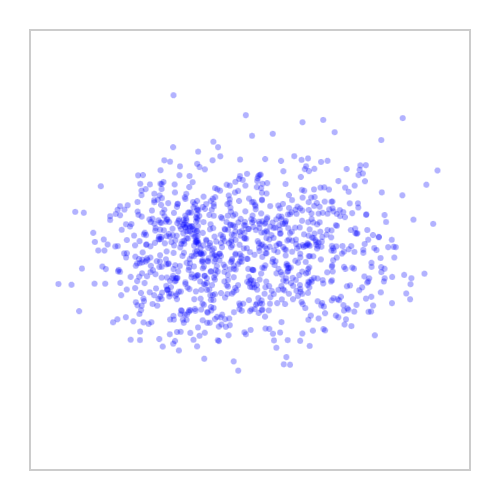}}
		\caption{\emph{Balanced and Nearby} distribution.}
		\label{fig:euclidean_1000,20,20,50,0.2,eucl6,uniform_square,1.5}
	\end{subfigure}
	
	\caption{Performance of the different rules for $n = 20$ voters, $T = 50$ rounds, $20$ alternatives in each round, with $f = 1.5$. The length of the bar represents the median across all the trials while the error bars represent the $25$th and $75$th percentile with the numeric text after the bar showing the mean value. The blue pictures in the right-most column illustrates the underlying distribution of voter locations.}
	\label{fig:param_setting_2}
\end{figure}

\begin{figure}
	\centering
	
	\begin{subfigure}{\textwidth}
		\centering
		\includegraphics[height=3.6cm]{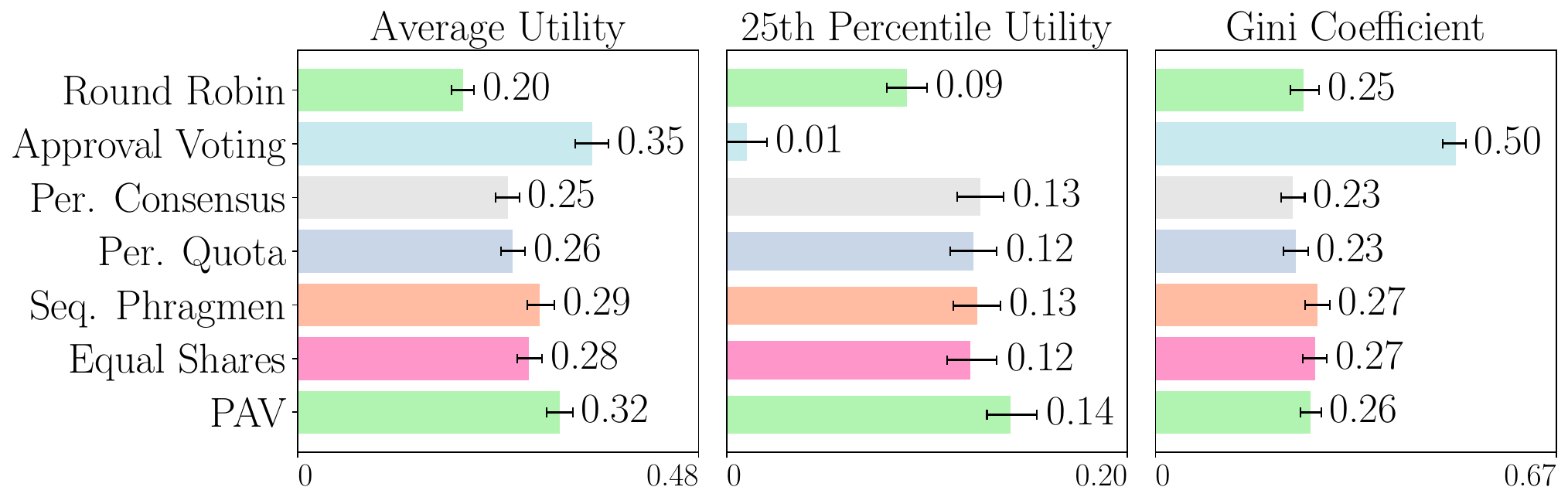}
		\iflatexml\else\:\raisebox{7pt}{\includegraphics[height=3.2cm]{Figures/distributions/restricted.png}}\fi
		\caption{\emph{Restricted} distribution.}
		\label{fig:euclidean_1000,50,40,50,0.2,eucl1,uniform_square,1.5}
	\end{subfigure}
	
	\vspace{\baselineskip}
	
	\begin{subfigure}{\textwidth}
		\centering
		\includegraphics[height=3.6cm]{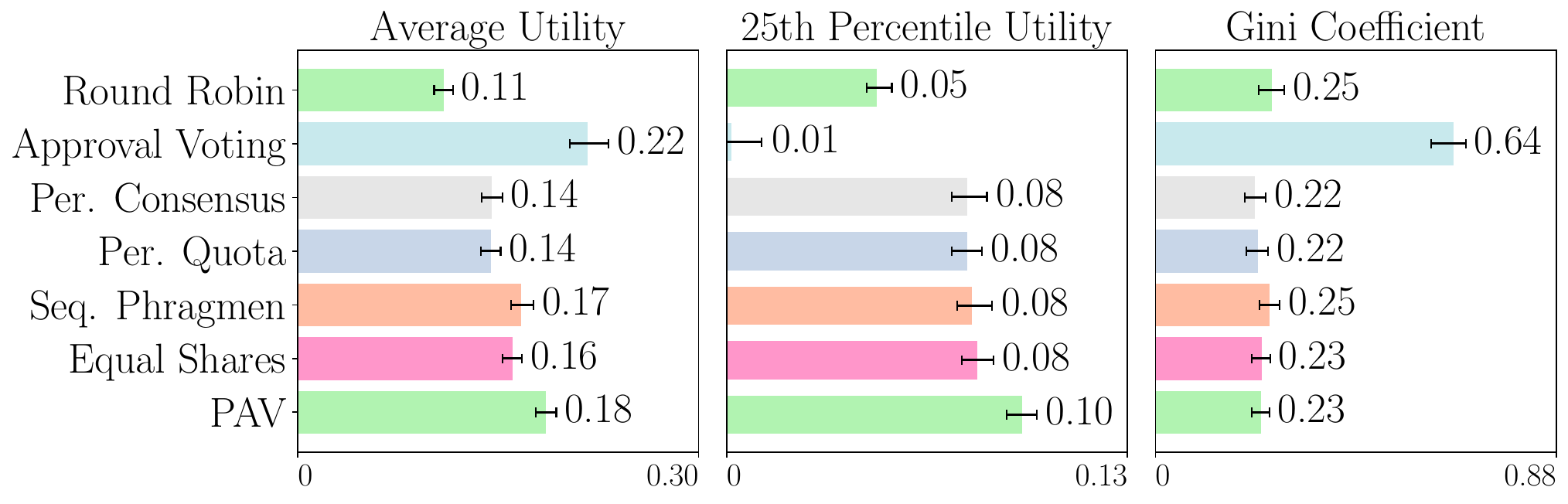}
		\iflatexml\else\:\raisebox{7pt}{\includegraphics[height=3.2cm]{Figures/distributions/many_groups.png}}\fi
		\caption{\emph{Many Groups} distribution.}
		\label{fig:euclidean_1000,50,40,50,0.2,eucl3,uniform_square,1.5}
	\end{subfigure}
	
	\vspace{\baselineskip}
	
	\begin{subfigure}{\textwidth}
		\centering
		\includegraphics[height=3.6cm]{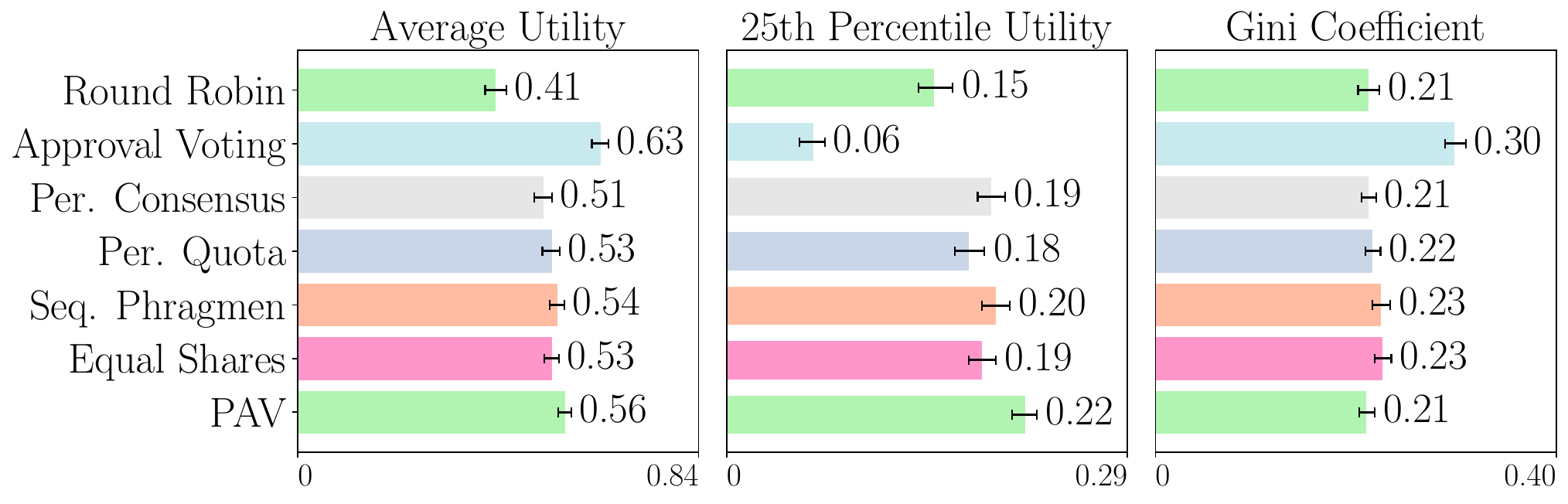}
		\iflatexml\else\:\raisebox{7pt}{\includegraphics[height=3.2cm]{Figures/distributions/unbalanced.png}}\fi
		\caption{\emph{Unbalanced} distribution.}
		\label{fig:euclidean_1000,50,40,50,0.1,eucl4,uniform_square,1.5}
	\end{subfigure}
	
	\vspace{\baselineskip}
	
	\begin{subfigure}{\textwidth}
		\centering
		\includegraphics[height=3.6cm]{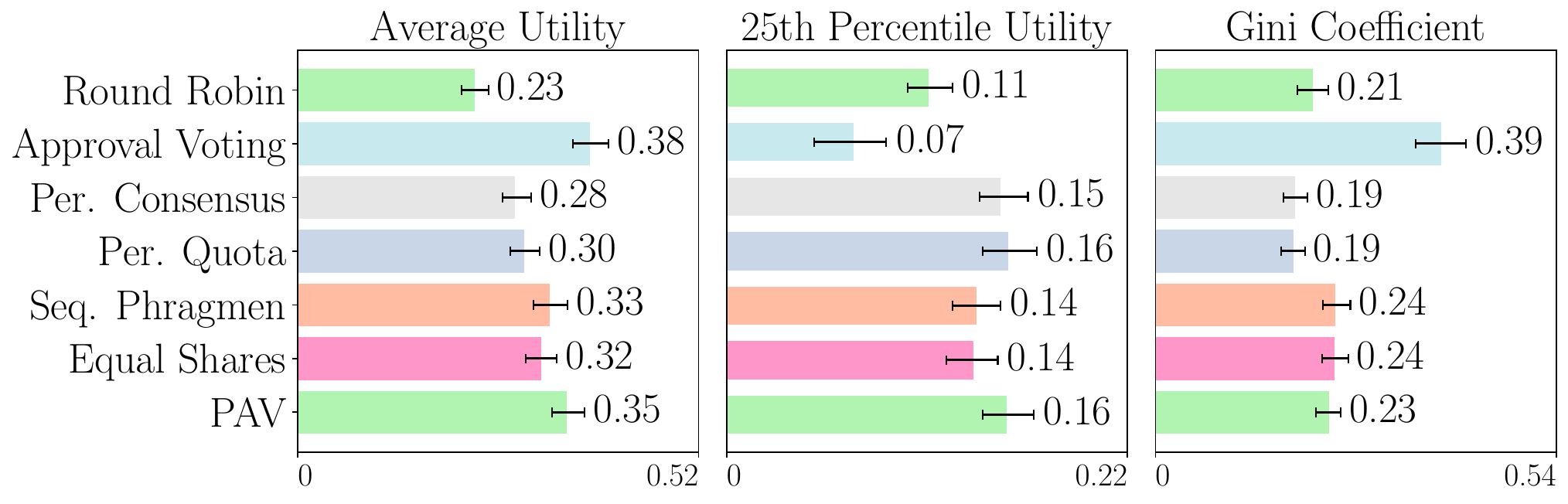}
		\iflatexml\else\:\raisebox{7pt}{\includegraphics[height=3.2cm]{Figures/distributions/balanced_nearyby.png}}\fi
		\caption{\emph{Balanced and Nearby} distribution.}
		\label{fig:euclidean_1000,50,40,50,0.2,eucl6,uniform_square,1.5}
	\end{subfigure}
	
	\caption{Performance of the different rules for $n = 50$ voters, $T = 50$ rounds, $40$ alternatives in each round, with $f = 1.5$. The length of the bar represents the median across all the trials while the error bars represent the $25$th and $75$th percentile with the numeric text after the bar showing the mean value. The blue pictures in the right-most column illustrates the underlying distribution of voter locations.}
	\label{fig:param_setting_3}
\end{figure}

\begin{figure}
	\centering
	
	\begin{subfigure}{\textwidth}
		\centering
		\includegraphics[height=3.6cm]{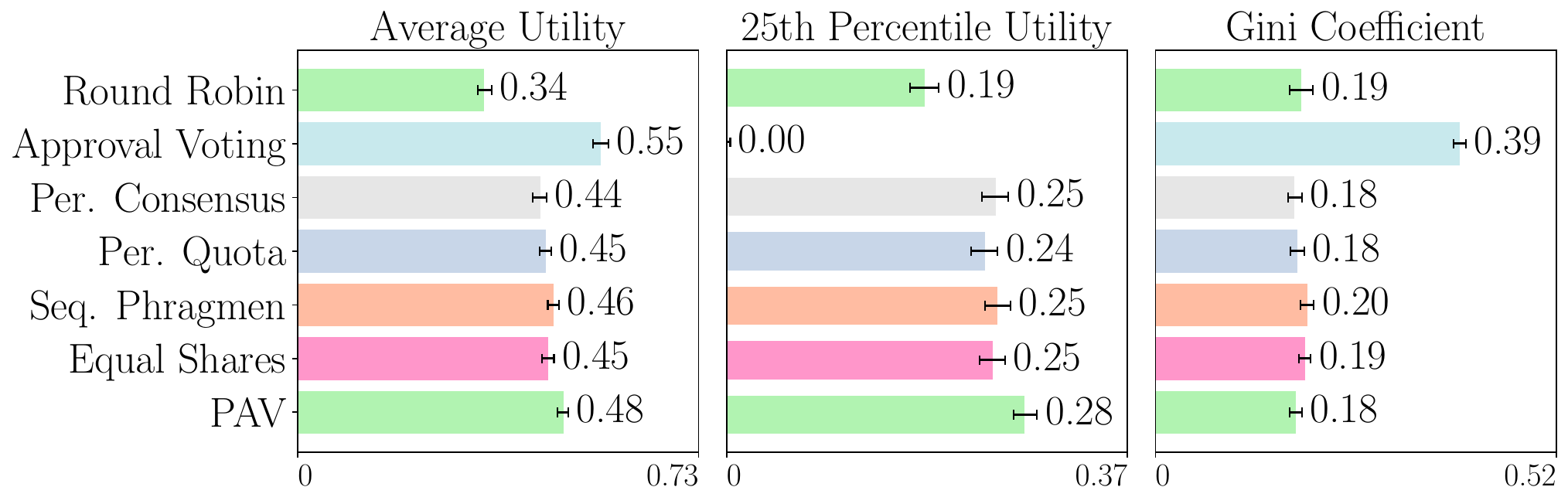}
		\iflatexml\else\:\raisebox{7pt}{\includegraphics[height=3.2cm]{Figures/distributions/restricted.png}}\fi
		\caption{\emph{Restricted} distribution.}
		\label{fig:euclidean_1000,50,40,50,0.2,eucl1,uniform_square,3}
	\end{subfigure}
	
	\vspace{\baselineskip}
	
	\begin{subfigure}{\textwidth}
		\centering
		\includegraphics[height=3.6cm]{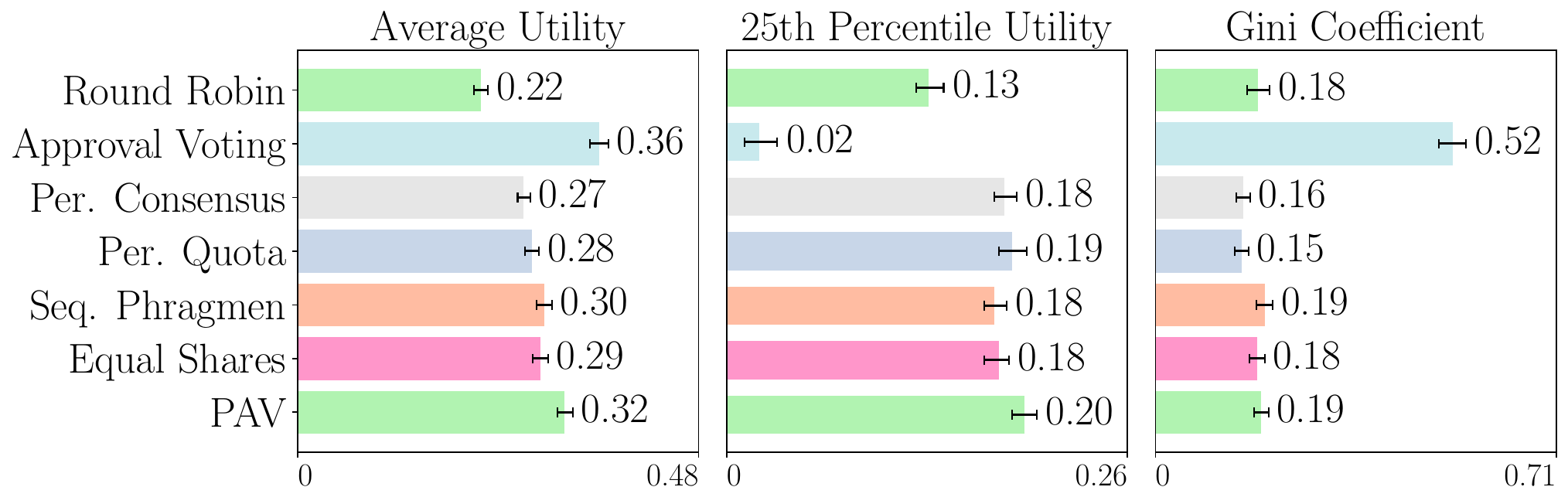}
		\iflatexml\else\:\raisebox{7pt}{\includegraphics[height=3.2cm]{Figures/distributions/many_groups.png}}\fi
		\caption{\emph{Many Groups} distribution.}
		\label{fig:euclidean_1000,50,40,50,0.2,eucl3,uniform_square,3}
	\end{subfigure}
	
	\vspace{\baselineskip}
	
	\begin{subfigure}{\textwidth}
		\centering
		\includegraphics[height=3.6cm]{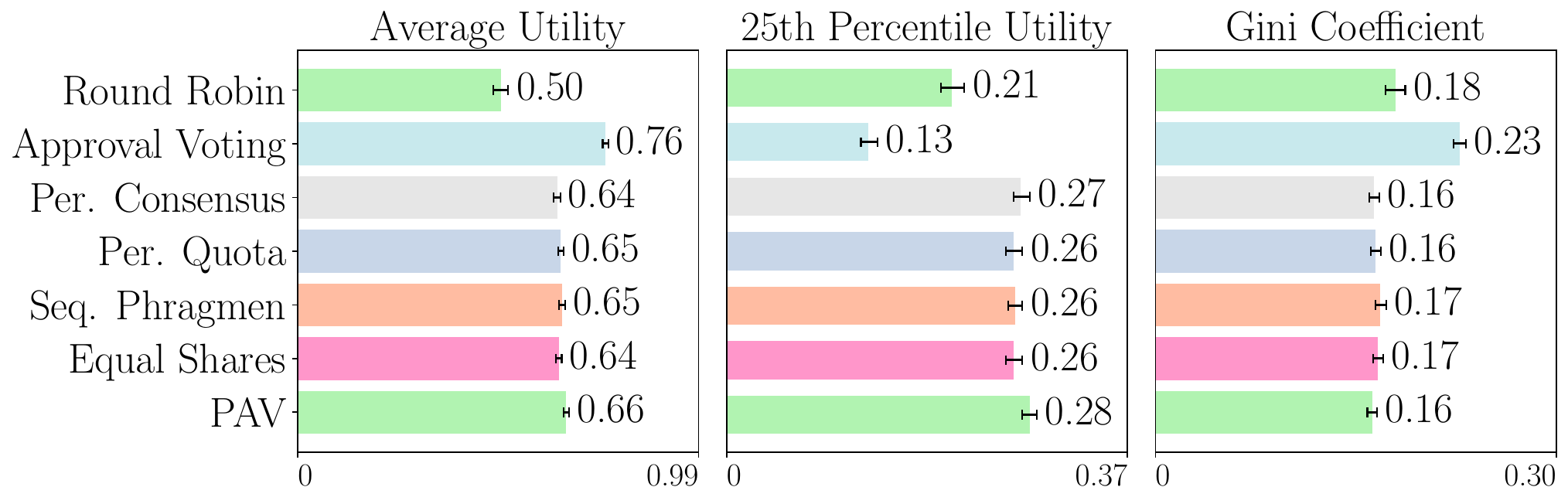}
		\iflatexml\else\:\raisebox{7pt}{\includegraphics[height=3.2cm]{Figures/distributions/unbalanced.png}}\fi
		\caption{\emph{Unbalanced} distribution.}
		\label{fig:euclidean_1000,50,40,50,0.1,eucl4,uniform_square,3}
	\end{subfigure}
	
	\vspace{\baselineskip}
	
	\begin{subfigure}{\textwidth}
		\centering
		\includegraphics[height=3.6cm]{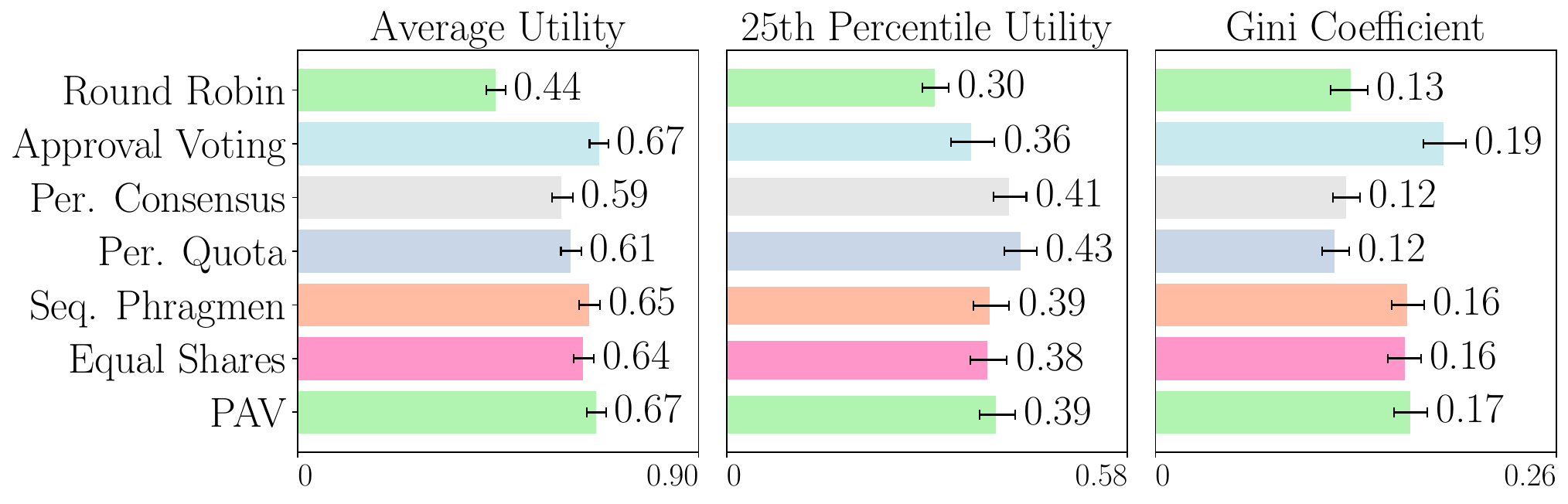}
		\iflatexml\else\:\raisebox{7pt}{\includegraphics[height=3.2cm]{Figures/distributions/balanced_nearyby.png}}\fi
		\caption{\emph{Balanced and Nearby} distribution.}
		\label{fig:euclidean_1000,50,40,50,0.2,eucl6,uniform_square,3}
	\end{subfigure}
	
	\caption{Performance of the different rules for $n = 50$ voters, $T = 50$ rounds, $40$ alternatives in each round, with $f = 3$. The length of the bar represents the median across all the trials while the error bars represent the $25$th and $75$th percentile with the numeric text after the bar showing the mean value. The blue pictures in the right-most column illustrates the underlying distribution of voter locations.}
	\label{fig:param_setting_4}
\end{figure}

Considering each metric, we make the following observations:

\begin{itemize}
	\item \emph{Average Satisfaction}: 
	In regard to this metric, a clear separation emerges between the voting methods across all the distributions. Approval Voting gives the highest average satisfaction, which is expected (it is easy to see that it always selects the decision sequence maximizing this metric). On the other hand, Round Robin performs the worst among the rules for all parameter settings and all distributions. Among the rules aiming for proportional outcomes, we see the following rank order for almost all parameter settings and distributions: PAV $>$ Sequential Phragm\'en $>$ MES $>$ Perpetual Quota $>$ Perpetual Consensus. The differences between adjacent rules in this ordering tend to be small. We see some exceptions to this ordering for the Unbalanced distribution, where Perpetual Quota does better than usual, but is still worse than PAV. The advantage of PAV over other rules is particularly apparent in the Many Groups distributions.
	
	\item \emph{Utility of the 25th Percentile}: 
	Here, we find that Approval Voting consistently performed poorly, often resulting in almost no decisions being approved by the bottom $25\%$ of the voters, leading to a score of close to 0 for this metric under the Restricted and the Many Groups distributions. Round Robin performed better than Approval Voting but still much worse than the other rules.
	
	Among those other rules, PAV generally performed the best. However, in the \emph{Balanced and Nearby} distribution, Perpetual Quota outperforms PAV for several parameter setting. Consensus performs similarly well for this distribution, but Sequential Phragmén and MES perform notably worse. For the other three distributions, the proportional rules all perform similarly on this metric, with Sequential Phragm\'en showing the best values for several parameter settings.
	
	\item \emph{Gini Coefficient}: 
	Across distributions, Approval Voting consistently produces the outcomes with the most inequality between voters, as measured by the Gini coefficient. That means that some voters are happy with many decisions, and others are happy with few decisions. All other rules produce outcomes with a more equal distribution of utility. Among the proportional rules, there is no consistent ranking between rules in their performance on this metric, and their Gini coefficients are often similar. Under the Balanced and Nearby distribution, however, Perpetual Consensus and Perpetual Quota have noticeably lower Gini coefficient. Depending on the parameters and distribution, the Round Robin method performs either slightly worse or comparably to the more sophisticated proportional rules -- recalling that, however, the average utility under this method is lower.
\end{itemize}

\subsection{Political Data}
\label{sec:experiments:political}

In addition to synthetic data, we evaluated the rules on data from U.S. political general elections which we collected from websites of local governments. The resulting dataset is available at \url{https://osf.io/t6p7s/}. In the elections we collected, voters elect candidates to various federal, state, and local political offices and express opinions on yes/no ballot initiatives. Using public anonymized \emph{Cast Vote Records (CVR)} data we can see each voter's votes on all these issues simultaneously.
We collected instances from 15 counties in California and Colorado from 2020 and 2022.  In these elections, voters usually only vote for a single alternative for each issue (rather than using approval). However, where several members needed to be elected to a board, voters are usually allowed to vote for several candidates (but no more votes than there are open positions). We reinterpreted such issues to be about the election of a single candidate, and interpreted the votes as forming an approval set. We did this to make the dataset fit better with our formal setup.
The voting rule actually used for these elections can be seen as equivalent to Approval Voting (i.e., in each round select the alternative with the highest vote count), one of the baseline rules included in our simulations.

Each county is divided into smaller voting districts, giving rise to a large number of decision instances. For example, in Colorado, we collected $169$ instances from $14$ counties. Each instance corresponds to all the voters that were issued a particular ballot type (usually, one instance per district of the county, since the candidates and offices differ between districts). We collected 41 instances from Shasta County, California, (one for each voting district) with each instance on average having 32 issues and 1\,682 voters. 
We ignored instances with more than $10\,000$ voters because the computation of PAV did not finish in reasonable time. 

\begin{figure}[t]
	\centering
	\begin{subfigure}{\textwidth}
		\centering
		\includegraphics[height=3.6cm]{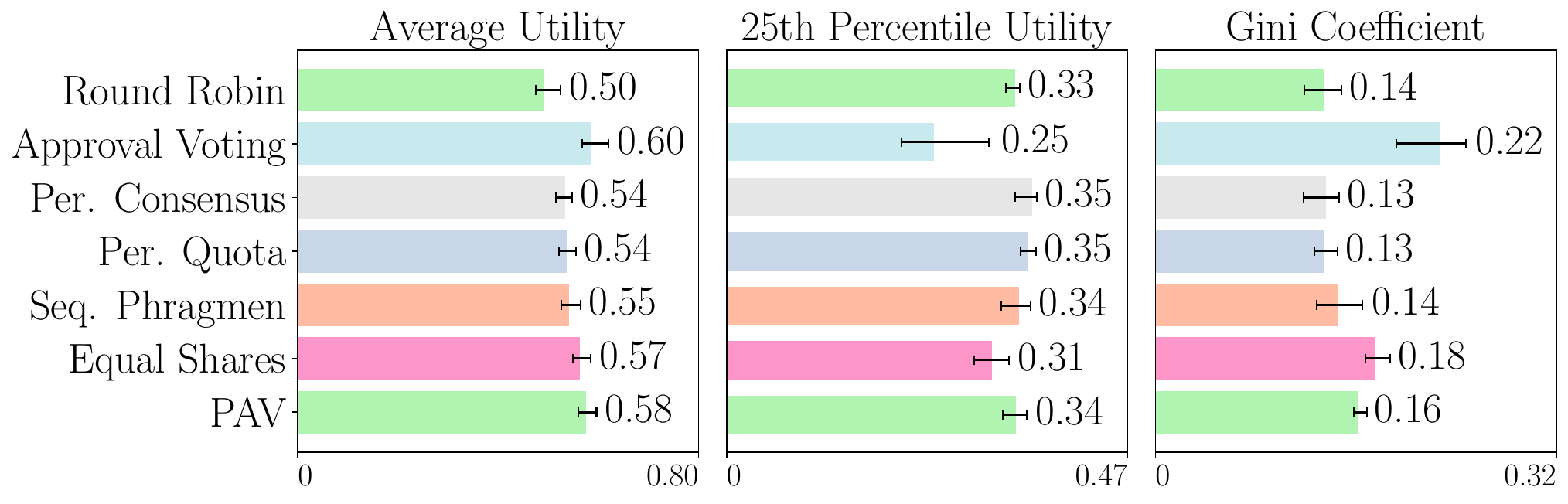}
		\caption{2022 General Elections in Shasta County, California}
	\end{subfigure}
	
	\medskip
	
	\begin{subfigure}{\textwidth}
		\centering
		\includegraphics[height=3.6cm]{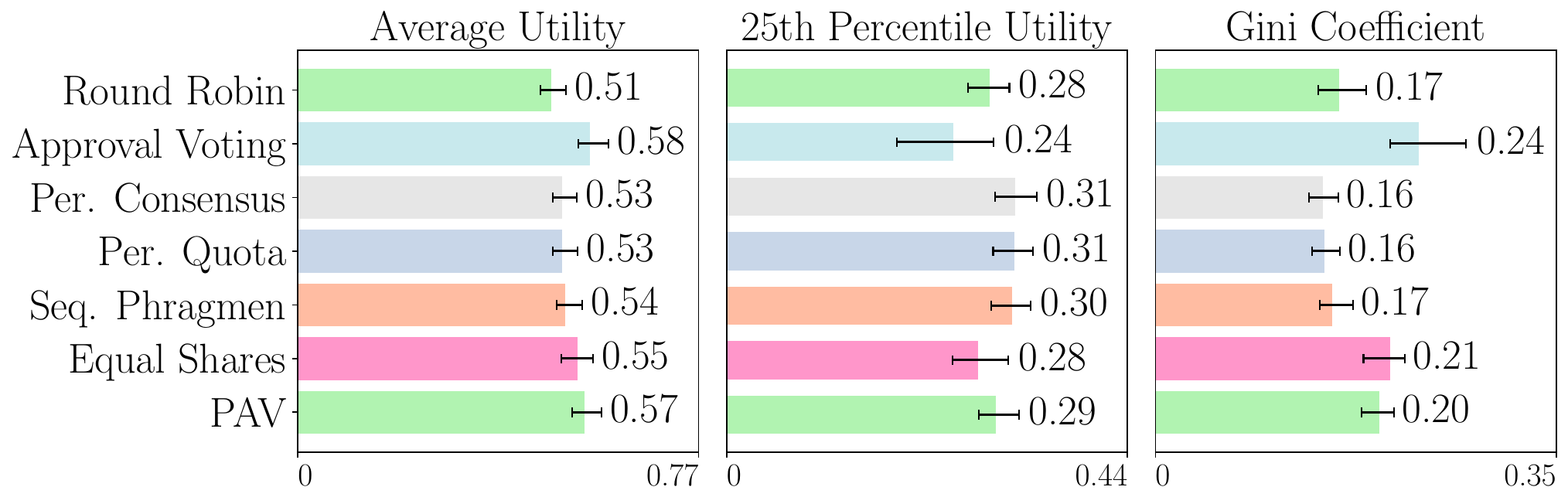}
		\caption{2022 General Elections in Colorado across 14 counties}
	\end{subfigure}
	\caption{Performance of the rules on the political dataset.
    The length of the bar shows the median across all voting districts, the error bars give the 25th and 75th percentile, and the text gives the mean.}
	\label{fig:political_data}
\end{figure}

The results are shown in  \Cref{fig:political_data} separately for Shasta County, Calif., and the 14 Colorado counties in our dataset. We see that Approval Voting (AV), the method that is actually used to determine winners, gives a higher Gini coefficient and lower utility at the 25th percentile compared to all the other rules. Strikingly, PAV leads to a more equal utility distribution while also having high average utility, with the difference to AV being almost negligible.
This suggests that proportional rules could be used to produce collections of decisions that are more consensual. We note, however, that much additional analysis would be appropriate, to understand the precise behavior of the rules on political instances.

\subsection{Learning Preferences from the Moral Machine}
\label{sec:experiments:learning}

\emph{Virtual democracy} \citep{noothigattu2018voting,freedman2020adapting,mohsin2021making} is a proposal to automate decision making by learning models of preferences of individual users and using predicted preferences as inputs to a voting rule.
This can be particularly useful when preference elicitation costs are high, decisions need to be made in real time, or a very large number of decisions need to be taken. 
A common approach \citep{noothigattu2018voting,wang2019privacy} for virtual democracy is to \emph{average} the learned model parameters across voters to obtain a single model, and effectively treat the predictions of the averaged model as the result of an aggregation method.  
\citet{moralmachinetyranny} have shown that such averaging may lead to non-proportional outcomes that underweight minority preferences. We hope to alleviate these fairness concerns by querying each individual model and then aggregating their outputs using voting rules. In particular, we will use the proportional aggregation methods that we have studied, and study their effect as compared to non-proportional rules and as compared to the traditional Machine Learning approach of learning a single preference model on a combined dataset.

\iflatexml
\begin{figure}
	\centering
	\includegraphics[width=\linewidth]{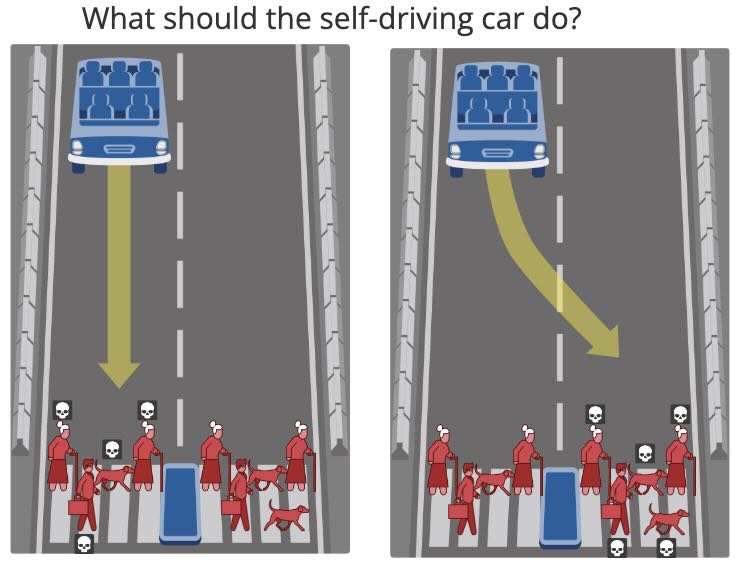}
	\caption{\href{https://www.moralmachine.net/}{moralmachine.net}}
	\label{fig:screenshot-moral-machine}
\end{figure}
\else
\begin{wrapstuff}[r,type=figure,width=5cm]
	\centering
	\includegraphics[width=\linewidth]{Figures/screenshot-moral-machine.jpg}
	\vspace{-18pt}
	\caption{\href{https://www.moralmachine.net/}{moralmachine.net}}
	\label{fig:screenshot-moral-machine}
\end{wrapstuff}
\fi
Following the work of \citet{noothigattu2018voting}, to empirically test our work on a dataset that has structured features which allow preference learning, we consider virtual democracy applied to the \emph{moral machine} \citep{awad2018moral}. This experiment is a modern take on ethical ``trolley problems'' and involves decisions that a self-driving car might face. Users were asked to express preferences in instances where a self-driving car must either swerve or stay in the lane, with both choices leading to injuring a different group of people. These choices can be seen as alternatives that can be described by a structured feature vector. 
Several million pairwise comparison responses are available in a public dataset.

\citet{noothigattu2018voting} learn a model predicting the preferences of each respondent. However, there are only 13 pairwise comparisons per respondent, meaning that such individual models have low accuracy. Instead, we partition the respondents by their country, and learn one model for each country. \citet{kim2018computational} found that grouping respondents from each country leads to higher accuracy on the moral machine dataset (perhaps due to cultural similarities). 

We limit ourselves to the 197 countries for which the dataset contains over 100 samples. To learn a preference model for each country, we use the Plackett--Luce (PL) model \citep{plackett1975analysis,duncan1959individual} which is a \emph{random utility model} appropriate for social choice preference learning \citep{azari2012random}. 
As a baseline, we train a \emph{combined model} on respondents from all countries, using $100$ samples from each country for a balanced representation in the training data. 

We produce 100 decision rounds together with 100 alternatives for each round. We specifically generate alternatives that feature high disagreement (since sampling alternatives uniformly at random typically induces very similar preferences across countries). We let each country approve the (predicted) best of these alternatives.
We then use our voting rules to compute decision outcomes using the country models as voters. 
To compare the voting rules to the performance of the combined model, we pick the alternative assigned the highest utility by the combined model as its decision for the round.

\begin{figure}[t]
	\centering
	\includegraphics[height=3.6cm]{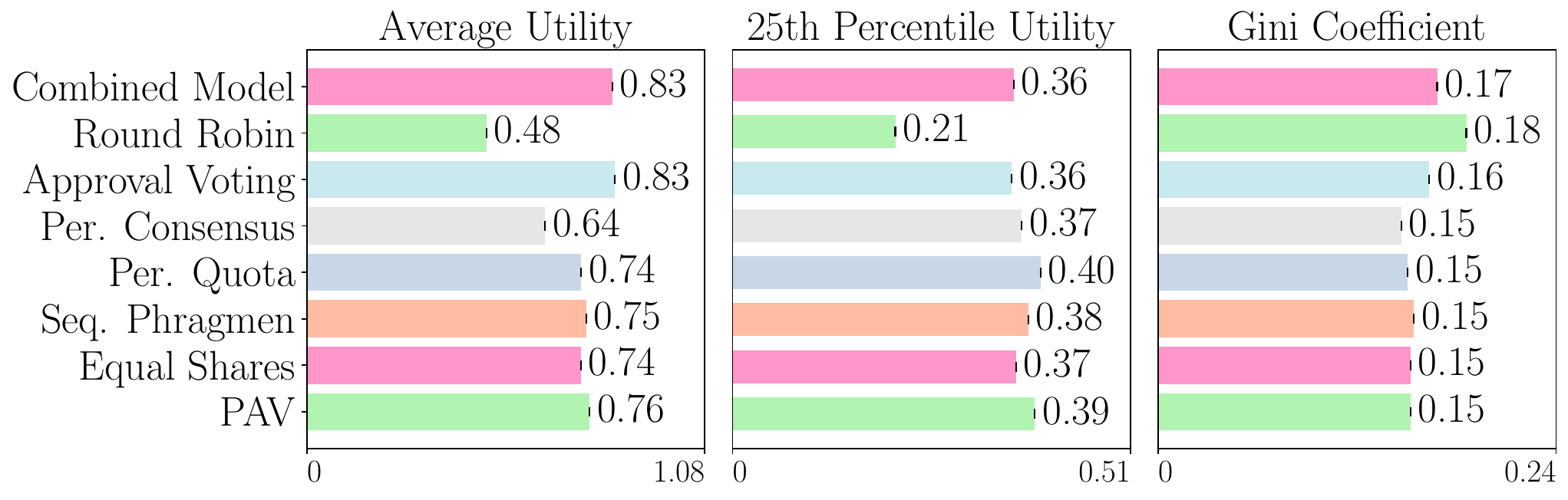}
\caption{Performance of the different rules on the Moral Machine Dataset, based on high-disagreement alternatives, for $T = 100$  rounds with each round having $100$ alternatives and each voter approving their top $1$ alternative.}
\label{fig:moral_machine_example}
\end{figure}

We present our results in \Cref{fig:moral_machine_example} (with results using other parameter settings shown in \Cref{app:moral-machine-results}). We find it striking how Approval Voting (AV) and the Combined Model attain almost identical values on each metric, and how these are quite different from the values obtained by the proportional rules. Indeed, in the experiment, AV and the Combined Model choose the same decision in 84\% of rounds, while both agree with the five proportional rules less frequently. Notably, the five proportional rules all feature a much smaller Gini coefficient and a higher satisfaction at the 25th percentile. 
The similar performance between AV and the Combined Model suggests that the Combined Model exhibits a bias towards plurality and majority opinions. This contradicts a possible hope one might have had that, because the Combined Model is trained based on an equal number of samples from each country, it will ``merge'' their views roughly proportionally.

We see the results from this small experiment as a potential starting point for a larger research program that studies how a global preference model trained on preferences of diverse agents makes decisions compared to aggregation rules that explicitly take into account each individual's preferences. Note that the results do not indicate such models can be used in applications, or to automate moral reasoning, which would require more rigorous testing and ethical considerations. 

\section{Discussion and Future Work}
\label{sec:conclusion}

\paragraph{Extensions}
Our model can be extended to make it compatible with more real-world applications. Examples include allowing weighting issues by importance \citep{page2020electingexecutive}, to allow for dependencies between issues \citep{propinterdependentissues}, or to allow voters to specify utilities or to rank alternatives instead of approvals \citep{peters2021proportional}. The latter may be important as some issues may be more critical to certain minority groups. 

\paragraph{Open Problems}
We leave some theoretical open problems for future work, notably whether an online rule can satisfy Weak EJR -- a negative result may be easier to find in this setting than for multi-winner voting.
Also open is whether MES can always be completed to satisfy PJR when all approval sets are non-empty. More conceptually, are there stronger versions of (Weak) EJR for this setting that are still satisfiable? 
Can the concept of proportionality degree \citep{skowron2021proportionality} be adapted to our setting? What about FJR (Full Justified Representation), EJR+, or the core? \citet{elkind2024verifying} recently defined an analog of the Greedy Cohesive Rule \citep{peters2021proportional} for the sequential setting, which would be a natural starting point for studying FJR.

\paragraph{Strategic Issues}
\citet{peters2018proportionality} proved an impossibility theorem showing that no proportional multi-winner voting rule can be \emph{strategyproof}: voters may be able to get a better outcome by misrepresenting their preferences. For the special case of ``approval-based apportionment'' \citep{approvalbasedapportionment,airiau2023portioning}, the impossibility still holds \citep{delemazure2022strategyproofness,freeriding2023}. As this is also a special case of our model (when the set of alternatives and voter preferences are the same in each round), it follows that no proportional rule in our setting can be strategyproof. Note that the presentation order of issues can also affect the results of online rules, offering another potential avenue for manipulation.

\section*{Acknowledgments}
We thank J\'er\^ome Lang for his support and advice throughout this research project, and for proposing key ideas and directions. We thank Jan Maly and Jannik Peters for helpful feedback. This work was funded in part by the French government under management of Agence Nationale de la Recherche as part of the ``Investissements d'avenir'' program, reference ANR-19-P3IA-0001 (PRAIRIE 3IA Institute). 

\appendix

\section{More Axioms from the Literature}
\label{app:moreaxioms}

In this section, we consider other proportionality axioms that have been introduced in prior work by \citet{lacknermaly2023}, and discuss implication relationships among them (see \Cref{fig:hasse} in the main body). We also discuss Pareto efficiency.

\subsection{Quota for Closed Groups}
\label{app:closed-groups}

\citet{lacknermaly2023} proposed some relatively weak notions of proportionality in the perpetual voting framework. One of their axioms applies only to ``\emph{simple profiles}'' (where the set of alternatives and the approval sets are the same in each round, with every voter approving just a single alternative). Then ``simple proportionality'' requires that the number of appearances of alternatives in the output be proportional to the number of voters approving them. This axiom is easily seen to be implied by all the other proportionality axioms we consider. Thus, we focus on their axioms based on profiles in which there are ``\emph{closed groups}'': groups of voters that submit identical approval sets in every round and such that any voter outside the group approves a disjoint set of alternatives.

\begin{definition}[Closed Group]
\label{def:closed_group}
    A group $S \subseteq N$ is \emph{closed} if for every round $j \in\{1, \ldots, T\}$, it holds that $A^i_j = A^{i'}_j$ for all $i, i' \in S$ and that $A^i_j \cap A^{i'}_j = \emptyset$ for all $i \in S$ and $i' \in N \setminus S$.
\end{definition}
For example, in \Cref{fig:axiom-example}, the group $\{1,2\}$ is closed but the group $\{3,4\}$ is not closed because voters 3 and 4 do not have the same approval set in rounds 7 and 8. In \Cref{fig:MESfails}, the groups $\{1\}$ and $\{2, 3\}$ are closed. The set of all voters is always closed, but singleton sets of voters are not necessarily closed due to the disjointness condition.

\citet{lacknermaly2023} defined axioms which bound the number of rounds in which the decision should be one of the alternatives approved by a closed group.

\begin{definition}[Perpetual Lower Quota for Closed Groups]
\label{def:plqc}
We say that a decision sequence $D = (d_1, \dots, d_T)$ satisfies \emph{perpetual lower quota for closed groups}\footnote{\citet{lacknermaly2023} call this axiom just ``lower quota for closed groups''. We added the word ``perpetual'' to the name since we will define below a weaker version suitable for the offline setting.} if, for every $k$, $1 \le k \le T$, writing $D_{:k} = (d_1, \dots, d_k)$ for  the prefix of $D$ of length $k$, it holds for every voter $i \in N$ who is part of a closed group $S$ that $U^i_{D_{:k}} \ge \lfloor k \cdot \smash{\frac{\left|S\right|}{n}} \rfloor$.
\end{definition}

\citet{lacknermaly2023} also proposed perpetual \emph{upper quota} for closed groups, which is defined in the same way, except that the inequality at the end is ``$U^i_{D_{:k}} \le \lceil k \cdot \smash{\frac{\left|S\right|}{n}} \rceil$''. Because this is a different kind of property (avoiding over-representation rather than guaranteeing representation), and because none of our rules satisfy it, we will not study it in more detail.

Observe that  perpetual lower quota for closed groups contains a requirement for all prefixes of the decision sequence $D$. This makes the axiom ``perpetual''. This is natural and desirable if the decision-making process is online. However, in an offline setting, it is less natural, and so we can drop the perpetuality and define a weaker axiom that is more broadly applicable. 

\begin{definition}[Lower Quota for Closed Groups]
\label{def:lqc}
A decision sequence $D$ satisfies \emph{lower quota for closed groups} if for every closed group $S$, we have for every voter $i \in S$ that  $U^i_D \geq \lfloor T \cdot |S| / n \rfloor$.
\end{definition}

For an example, refer to \Cref{fig:MESfails} where lower quota for closed groups would demand that each voter in $S = \{1\}$ should approve at least $\lfloor 6 \cdot 1 / 3 \rfloor = 2$ decisions of the outcome, while each voter in $S' = \{2, 3\}$ should approve at least $\lfloor 6 \cdot 2 / 3 \rfloor = 4$ decisions of the outcome.

We now show that lower quota for closed groups is implied by Weak PJR, and therefore also by the stronger properties PJR and (Weak) EJR.

\begin{theorem}
    \label{thm:pjr_implies_lqc}
    Weak PJR implies lower quota for closed groups.
\end{theorem}
\begin{proof}
Let $D$ be a decision sequence that satisfies Weak PJR. 
Let $S \subseteq N$ be a closed group. Write $\ell = \lfloor T \cdot |S| / n \rfloor$. Then $|S| = (T \cdot |S| / n) \cdot \frac{n}{T} \ge \ell \cdot \frac{n}{T}$, and $S$ agrees in every round since it is a closed group. Thus, Weak PJR implies that there are $\ell$ rounds where at least one member of $S$ approves the decision of $D$. But because $S$ is closed, this means that in those $\ell$ rounds, \emph{every} member of $S$ approves the decision. Hence $D$ satisfies lower quota for closed groups.
\end{proof}

Similarly, an \emph{online} rule that satisfies PJR will thereby satisfy perpetual lower quota for closed groups.

\subsection{Priceability}
\label{app:priceability}

\citet{PeSk20} introduced an axiom called \emph{priceabilty} to ensure proportionality and (approximately) equal influence among voters in the setting of multi-winner voting. The essence of priceability lies in assigning each voter an equal budget of virtual money, which they can only spend on alternatives they approve. Thinking of the virtual money as voting power, the axiom ensures that each voter has approximately equal influence over the outcome. \citet{lacknermaly2023} extended priceability to the setting of perpetual voting and showed that Sequential Phragm\'en (which they call Perpetual Phragm\'en) satisfies perpetual priceability. We state their axiom below and discuss its relation to our axioms.

Priceability is defined based on the idea that it costs 1 unit of money to pay for the decision in 1 round. Then a ``price system'' is used to split this cost between voters that approve this decision.

\begin{definition}[Price System]
\label{def:price_system}
Suppose we are looking at the first $k$ rounds.%
\footnote{We use a variable $k$ instead of the total number of rounds $T$ since the definition of perpetual priceability concerns decision sequences that are shorter than the full time horizon $T$.}
A \emph{price system} is a pair $(B,\{p_j\}_{j \leq k})$ where $B\ge 0$ is the \emph{budget} that each voter starts with, and $p_j : N \times C_j \to [0,1]$, for each $j \in \{1, \ldots, k\}$, is a \emph{payment function} where $p_j(i, c)$ will indicate how much agent $i$ pays for candidate $c$ in round $j$.
We say that a decision sequence $D = (d_1, \ldots, d_k)$ of length $k$ is \emph{supported by the price system} $(B,\{p_j\}_{j \leq k})$ if
\begin{itemize}[leftmargin=30pt]
    \item [(P1)] $p_j(i, c)=0$ whenever $c \notin A^i_j$, i.e., no voter pays for non-approved alternatives.
    \item [(P2)] $\sum_{j=1}^k \sum_{c \in C_j} p_j(i, c) \le B$, i.e., voters cannot spend more than their budget.
    \item [(P3)] $\sum_{i \in N} p_j(i, d_j)=1$ for $j \in\{1, \ldots, k\}$, i.e., each decision of $D$ gets a total payment of $1$.
    \item [(P4)] $\sum_{i \in N} p_j(i, d)=0$ for $j \in\{1, \ldots, k\}$ and $d \neq d_j$, i.e., decisions not included in $D$ do not receive any payments.
\end{itemize}
\end{definition}

To define priceability for perpetual voting, \citet{lacknermaly2023} propose the concept of a \emph{minimal} price system which intuitively requires that no money is wasted at any time step.

\begin{definition} [Minimal Price System, \citealp{lacknermaly2023}]
Consider a decision sequence $D = (d_1, \ldots, d_k)$ of length $k$ supported by a price system $(B,\{p_j\}_{j \leq k})$ satisfying conditions (P1)-(P4) of \Cref{def:price_system}.
If $k = 0$ (so that $D = ()$), we say that the price system is \emph{minimal} if $B = 0$.
If $k \ge 1$, we say that the price system is minimal if it satisfies the following two conditions:
\begin{itemize}[leftmargin=30pt]
    \item [(P5)] there exists a minimal price system $\big(B^*,\{p_j^*\}_{j \le k-1}\big)$ that supports $(d_1, \ldots, d_{k-1})$
    \item [(P6)] there are no $B^{\prime}, d_k^{\prime}$ and $p_k^{\prime}$ such that $B^* \leq B^{\prime}<B$ and $\big(B^{\prime},\{p_j\}_{j \leq k-1} \cup\left\{p_k^{\prime}\right\}\big)$ is a price system supporting $(d_1, \ldots, d_{k-1}, d_k^{\prime})$.
\end{itemize}
\end{definition}

We can now define perpetual priceability.
\begin{definition} [Perpetual Priceability, \citealp{lacknermaly2023}]
A decision sequence $D$ satisfies perpetual priceability if there exists a minimal price system that supports $D$.
\end{definition}

With the above definition in place, we will now show that perpetual priceability implies Weak PJR. (Similarly, in multi-winner voting, priceability implies Weak PJR \citep{PeSk20}.) The proof below follows the same style as a proof by \citet{lacknermaly2023} that perpetual priceability implies perpetual lower quota for closed groups (\Cref{def:plqc}). 

\begin{theorem}
    \label{thm:priceability_implies_PJR}
    For decision sequences of length $T$, Perpetual Priceability implies Weak PJR.
\end{theorem}
\begin{proof}
Let $D = (d_1, d_2, \ldots, d_T)$ be a decision sequence of length $T$ that satisfies perpetual priceability. Assume for a contradiction that $D$ violates Weak PJR. Thus, there is $\ell \in \mathbb N$ and a group $S \subseteq N$ which agrees in all the $T$ rounds, has size $|S| = s \ge \ell \cdot  \frac{n}{T}$, but there are fewer than $\ell$ rounds in which at least one member of $S$ approves the decision in $D$. 

Let $\big(B,\{p_j\}_{j \leq T}\big)$ be a minimal price system that supports $D$. Note that $B \geq T / n$, as otherwise it would not be possible to pay for $T$ decisions. 
Assume first that $B = T / n$. Then no budget is left after round $T$, and thus every voter has spent their entire budget $B$. Thus, the total payment of voters in $S$ is $|S| \cdot \frac{T}{n} \ge \ell$. 
Since each decision is made using 1 unit, in total $S$ must have paid for the decision in at least $\ell$ rounds. As a voter only pays for an alternative they approve, it follows that in each of those $\ell$ rounds, at least one voter in $S$ approves the decision. Hence, Weak PJR is satisfied. 

So assume $B > \frac{T}{n}$. Because at most $\ell - 1$  decisions in $D$ are approved by some member of $S$, the members of $S$ have spent at most $\ell - 1$ units until after round $T$. Since they start out with a total budget of $sB$, we see that the voters in $S$ have a remaining budget of at least
\begin{align}
sB - (\ell - 1)
\ge
sB - (s \tfrac Tn - 1)
> s \tfrac Tn - (s \tfrac Tn - 1)
\ge 1
\label{eq:priceability_remaining_budget}
\end{align}
Hence there is an $\epsilon > 0$ such that the remaining budget of the voters in $S$ is $1+\epsilon$. By the definition of minimal price systems, for every $1 \le r \le T$, there is a price system $ (B_{r-1}, \{p_j^{(r-1)} \}_{j \leq r-1} )$ that witnesses the minimality of the price system $ (B_r, \{p_j^{(r)} \}_{j \leq r} )$, where $ (B_T, \{p_j^{(T)} \}_{j \le T} )= (B, \{p_j \}_{j \le T} )$.

We claim that $B_1<B$. Observe that in the first round $B_1 \le 1 / s$ must hold, as with a budget of $1 / s$ the voters in $S$ can together afford one of the alternatives that they jointly approve; thus $B_1 > 1/s$ would contradict minimality in round 1. Moreover, we can assume that $\ell \geq 1$ as otherwise Weak PJR is trivially satisfied. Thus $s \ge \frac n T$, and hence $n / s \le T$. Also, by assumption, $B>T / n$. Putting all of this together, we get
\[
B_1 \le \tfrac{1}{ s}=\tfrac1n \cdot \tfrac{n}{s} \le \tfrac{T}{n}<B .
\]

Now let $r^*$ be the largest index $r$ for which $B_r < B$ (recalling that at the end, $B_T = B$). We claim that there is a $B^{\prime}$ with $B_{r^*} \leq B^{\prime}<B_{r^*+1} = B$ such that there is an alternative $a^{\prime} \in C_{r^* + 1}$ and $p_{r^{*}+1}^{\prime}$ such that $ (B^{\prime}, \{p_j^{(r^*)} \}_{j \leq r^*} \cup \{p_{r^* + 1}^{\prime} \}) $ is a price system supporting the decision sequence $ (d_1, \ldots, d_{r^*}, a^{\prime} )$. This would be a contradiction to the minimality of $ (B_{r^* + 1}, \{p_j^{(r^* + 1)} \}_{j \leq r^*+1} )$.

Let $a^{\prime}$ be an alternative approved by everyone in $S$ in round $r^*+1$. Furthermore, let 
\[
 B^{\prime}= \max  (B - \epsilon / s, B_{r^*} ).
\]
Observe that $B_{r^*} \leq B^{\prime} < B$. 
Now, we can define $p_{r^* + 1}^{\prime}$ such that $\sum_{i \in S} p_{r^*+1}^{\prime} (i, a^{\prime} )=1$ and $p_{r^*+1}^{\prime} (i^{\prime}, c )=0$ whenever $i^{\prime} \notin S$ or $c \neq a^{\prime}$. This is possible because the voters in $S$ have at least a budget of $1$ in round $r^*+1$ (from  \eqref{eq:priceability_remaining_budget}, recalling that $B_{r^* + 1} = B$).
Hence, this is a price system that supports $ (d_1, \ldots, d_{T-1}, a^{\prime} )$ which contradicts the minimality of $ (B, \{p_j \}_{j \leq T} )$.
\end{proof}

As perpetual priceability implies Weak PJR, one may wonder whether it implies stronger axioms. Since Sequential Phragm\'en satisfies perpetual priceability \citep{lacknermaly2023} but fails Weak EJR (\Cref{thm:phragmen_fails_EJR}), it follows that perpetual priceability does not imply Weak EJR. The next result shows that it does not imply PJR either.

\begin{theorem}\label{thm:priceability_does_not_imply_some_rounds}
    Perpetual Priceability does not imply PJR.
\end{theorem}

\begin{figure}
	\centering
{\setlength{\tabcolsep}{2.5pt}
		\begin{tabular}{lcccccccc}
			\toprule
			Round \hspace{6pt} & 1 & 2 & 3 & 4 \\
			\midrule
			Voter 1 & $\{\alternativea\}$ & $\{\alternativeg\}$ & $\{\alternativeg\}$ & $\{\alternativeg\}$ \\
			Voter 2 & $\{\alternativeb\}$ & $\{\alternativeg\}$ & $\{\alternativeg\}$ & $\{\alternativeg\}$ \\
			Voter 3 & $\{\alternativec\}$ & $\{\alternativeg\}$ & $\{\alternativeg\}$ & $\{\alternativeg\}$ \\
			Voter 4 & $\{\alternatived\}$ & $\{\alternativeh\}$ & $\{\alternativeh\}$ & $\{\alternativeh\}$ \\
			Voter 5 & $\{\alternativee\}$ & $\{\alternativeh\}$ & $\{\alternativeh\}$ & $\{\alternativeh\}$ \\
			Voter 6 & $\{\alternativef\}$ & $\{\alternativeh\}$ & $\{\alternativeh\}$ & $\{\alternativeh\}$ \\
			\bottomrule
	\end{tabular}}
	\caption{Example instance where a rule may produce a perpetually priceable outcome which may not satisfy JR.}
    \label{fig:priceability_does_not_imply_some_rounds}
\end{figure}

\begin{proof}
Consider the decision instance shown in \Cref{fig:priceability_does_not_imply_some_rounds} with $T = 4$ rounds and $N = \{v_1, \ldots, v_6\}$ with each voter having a single-alternative approval set in each round. 
Consider the decision sequence $D = (\alternativea, \alternativeh, \alternativeh, \alternativeh)$. We show that it satisfies perpetual priceability but fails PJR.

In the first round, everyone has a disjoint approval set, so no matter which alternative is chosen in this round, one voter needs to pay its entire cost. Thus, any price system supporting the decision in the first round must have $B_1 = 1$. Setting $p_1(v_1, \alternativea) = 1$ and $p_1(i, c) = 0$ for all $i \neq v_1$ and all $c \neq \alternativea$, we obtain a minimal price system $(B_1, p_1)$ supporting $D' = (\alternativea)$.

Thus, everyone receives the budget of $1$ unit after round $1$ and only $v_1$ has spent all of its money while other voters have their entire budget remaining. In round $2$, $\alternativeh$ can be bought by each of its supporters paying $\frac13$ unit as $3\cdot \frac13 = p = 1$. Similarly, $\alternativeh$ can be brought in rounds $3$ and $4$ as well without requiring to increase the budget of anyone. This defines payment functions $p_2$, $p_3$, and $p_4$ for those rounds, and setting $B = 1$, we obtain a minimal price system $(B, \{p_j\}_{j \le 4})$ supporting $D$. Thus, $D$ satisfies perpetual priceability. 

However, $D$ fails PJR: Consider the coalition $S = \{v_2, v_3\}$ which agrees in the last $k = 3$ rounds. For $\ell = 1$, we have $|S| \ge \ell \cdot \frac{n}{k} = \frac63 = 2$. Thus, PJR requires that at least $1$ decision in the outcome must be approved by $S$. However, no one in $S$ approves either $\alternativea$ or $\alternativeh$, and thus $D$ fails PJR. Hence, perpetual priceability does not imply PJR. 
\end{proof}

\begin{remark}
	In the example in \Cref{thm:priceability_does_not_imply_some_rounds}, we used a requirement of $\ell = 1$ rounds, so the proof of \Cref{thm:priceability_does_not_imply_some_rounds} shows that perpetual priceability does not even imply JR.
\end{remark}

Note that perpetual priceability implies PJR for coalitions that agree ``initially'', that is they agree in an initial set of rounds $R^* = \{1, \dots, k\}$ (since by definition of perpetual priceability, the decision sequence restricted to the first $k$ decisions is itself perpetually priceable, and therefore this restricted sequence satisfies Weak PJR, which is enough to imply PJR with initial agreement). Note that the example in \Cref{thm:priceability_does_not_imply_some_rounds} does not feature initial agreement.

In future work, it would be interesting to define a notion of priceability that applies to our setting without an implicit assumption of being online (that is, without necessarily requiring that all prefixes of the decision sequence are themselves priceable).

\subsection{Pareto Efficiency}
\label{app:pareto}

Pareto efficiency forbids that we use our available resources suboptimally. Formally, an outcome is Pareto efficient when it is impossible to change the outcome so as to make some individual better off without making anyone else worse off.

\begin{definition}[Pareto Efficiency]
    \label{def:pareto_efficiency}
    A decision sequence $D$ is \emph{Pareto efficient} if there does not exist another decision sequence $D' \in C_1 \times C_2 \times \cdots \times C_T$ such that for all $i \in N$, $U^i_{D'} \ge U^i_D$ and there is some voter $i \in N$ such that $U^i_{D'} > U^i_D$.
\end{definition}

Unlike our other axioms, Pareto efficiency is not about proportionality (since a Pareto efficient outcome can be very non-proportional). However, it is interesting to see if any of our rules satisfy Pareto efficiency, and therefore combine an efficiency guarantee with proportionality. Since it maximizes a function of voter utilities, PAV satisfies Pareto efficiency, while also satisfying our strongest proportionality notions. On the other hand, MES and Sequential Phragm\'en fail Pareto efficiency. This makes PAV a compelling rule for multi-issue collective decision making, when an NP-hard offline method is acceptable. (The outcome of Local-Search PAV need not be Pareto-efficient.)

\begin{theorem}
    \label{thm:PAV_satisfies_PE}
    The decision sequence produced by PAV is Pareto efficient. 
\end{theorem}

\begin{proof}
    Suppose for a contradiction that the decision sequence $D$ produced by PAV is not Pareto efficient. Then there exists another decision sequence say $D'$ such that for all $i \in N$, $U^i_{D'} \ge U^i_D$ and there is some voter $v \in N$ such that $U^v_{D'} > U^v_D$. Thus, the $\operatorname{PAV-score}$ of $D'$ is
    \begin{align*}
    \operatorname{PAV-score}(D') &= 
    \sum_{i\in N} \left ( 1 + \frac12 + \frac13 + \dots + \frac1{U^i_{D'}} \right ) 
    > \operatorname{PAV-score}(D)
    \end{align*}
    which contradicts that PAV selects the decision sequence which maximizes the $\operatorname{PAV-score}$.
\end{proof}

\begin{theorem}\label{thm:MES_Phragmen_pareto_inefficient}
    The decision sequence produced by MES and Sequential Phragm\'en may not be Pareto efficient. 
\end{theorem}

\begin{figure}
	\centering
	 \scalebox{0.95}
    {\setlength{\tabcolsep}{4pt}
		\begin{tabular}{lccccccc}
			\toprule
			Rounds & 1 - 2 & 3 - 7 \\
			\midrule
			Voter 1 & $\{\alternativea\}$ &  $\{\alternativec\}$ \\
			Voter 2 & $\{\alternativea\}$ &  $\{\alternatived\}$ \\
			Voters 3, 4, 5, 6, 7 & $\{\alternativeb\}$ &  $\{\alternativeb\}$ \\
\bottomrule
	\end{tabular}}
	\caption{Example instance where both MES and Sequential Phragm\'en produce Pareto inefficient outcomes.}
    \label{fig:not_pareto_efficient}
\end{figure}

\begin{proof}
\Cref{fig:not_pareto_efficient} shows an example where both MES and Sequential Phragm\'en fail Pareto efficiency, with $T = 7$ rounds and $n = 7$ voters. Here, voters $1$ and $2$ agree in the first $2$ rounds on $\alternativea$, but neither of them agrees with anyone else in other rounds. Meanwhile, voters $3, 4, 5, 6, 7$ agree in all rounds on $\alternativeb$. 
Both Sequential Phragm\'en and MES would produce $D = (\alternativeb, \alternativeb, \alternativeb, \alternativeb, \alternativeb, \alternativec, \alternatived)$ (or breaking ties differently, with $\alternatived$ in round 6 and $\alternativec$ in round 7) but the decision sequence $D' = (\alternativea, \alternativea, \alternativeb, \alternativeb, \alternativeb, \alternativeb, \alternativeb)$ Pareto dominates $D$ as $U^1_{D'} = 2 > U^1_D = 1$ while $U^i_{D'} \ge U^i_D$ for all $i \in N$. Intuitively, it is better to satisfy voters $1$ and $2$ using the efficient choice $\alternativea$ which is available in the first two rounds, rather than using the less efficient alternatives $\alternativec$ and $\alternatived$. But due to their (semi-)online nature, neither MES nor Sequential Phragm\'en can know that an alternative as good as $\alternativea$ will not be available in later rounds.
\end{proof}

In the example in the proof of \Cref{thm:MES_Phragmen_pareto_inefficient}, Offline MES can select the same Pareto inefficient decision sequence, so it also fails Pareto efficiency.

  \section{Omitted Proofs from the Main Text}
\label{sec:omittedproofs}

\subsection{Impossibility of Stronger Guarantees: More Details}
\label{app:infeasible-axioms}

Here we give proofs and additional results related to the discussion in \Cref{sec:hardness}.
The results are based on a counterexample construction formally described in the following result. It has the following structure: fix some coalition size $s \le n$ of interest. In the first $k$ rounds, for each of the $\binom{n}{s}$ subsets $S\subseteq N$ of size $s$, there is an alternative $a_S$ with all voters in $S$ approving $a_S$. Thus, all these coalitions agree in $k$ rounds, and thus the impossible axioms we will consider guarantee all of these coalitions some representation. However, $k$ rounds is not enough to ``hit'' enough of these coalitions, and in the remaining $T - k$ rounds, voters are in complete disagreement (i.e., the approval sets of any two voters are disjoint) and so there are not enough rounds to make up the deficit accrued in the first $k$ rounds.

\begin{theorem}
	\label{thm:counterexample-construction}
	Let $\epsilon > 0$. Choose some $k > \lceil \frac{1-\epsilon}{\epsilon} \rceil$ and some $T > k$. Then for sufficiently large $n$, there exists an instance with $n$ voters  and time horizon $T$ such that for every decision sequence $D$, there is a coalition $S \subseteq N$ that agrees in $k$ rounds and has size $|S| \ge (1 - \epsilon) \cdot \frac{n}{k}$, yet none of the voters in $S$ approve any of the decisions of $D$.
\end{theorem}

\begin{proof}
	Pick $n \in \mathbb N$ such that $n \geq \frac{k(T+1)}{\epsilon k + \epsilon - 1}$. Note for later that (because $\epsilon k + \epsilon - 1 > 0$, since $k > \lceil \frac{1-\epsilon}{\epsilon} \rceil$), we have
	\[
	n(\epsilon k + \epsilon - 1) \geq k(T+1).
	\]
	Adding $nk$ to both sides, and rearranging,
	\[
	nk \geq n \cdot (1-\epsilon) \cdot k + n(1 - \epsilon) + k(T+1).
	\]
	Finally, dividing by $k$, we get
	\begin{equation}
		\label{eq:rearrange-n}
		n \geq n(1 - \epsilon) + \tfrac{n(1 - \epsilon)}{k} + T + 1.
	\end{equation}
	
	We now construct an instance with $n$ voters and time horizon $T$. Write $s = \lceil (1 - \epsilon) \cdot \frac{n}{k} \rceil$. In the first $k$ rounds, there are $\binom{n}{s}$ many alternatives, one for each coalition $S \subseteq N$ with $|S| = s$, that is approved by exactly the voters in $S$. In the remaining $T - k$ rounds, there are $n$ alternatives, each approved by a distinct voter (thus each voter has a singleton approval set in those rounds, and no two voters approve the same alternative).
	
	Let $D$ be any decision sequence. We first compute how many voters approve at least 1 decision in $D$. In the first $k$ rounds, the decisions in $D$ can satisfy at most $k \cdot s$ different voters. In the remaining $T - k$ rounds, at most 1 voter can be satisfied per round, so at most $T - k$ in total. Hence the number of satisfied voters is at most
	\begin{align*}
		k \cdot s + (T - k) &= k \cdot \lceil (1 - \epsilon) \cdot \tfrac{n}{k} \rceil + (T - k)\\
		&\le n \cdot (1 - \epsilon) + k + (T - k) \\
		&= n \cdot (1 - \epsilon) + T.
	\end{align*}
	Thus, the number of voters who approve none of the decisions in $D$ is at least
	\[
	n - (n \cdot (1-\epsilon) + T) \overset{\eqref{eq:rearrange-n}}{\ge} \tfrac{n(1 - \epsilon)}{k} + 1 \ge s.
	\]
	Take any set $S \subseteq N$ of exactly $s$ voters who approve none of the decisions in $D$. By construction of the instance, they agree in the first $k$ rounds, and $|S| = s = \lceil (1 - \epsilon) \cdot \frac{n}{k} \rceil$. Thus, $S$ satisfies the conditions promised in the theorem statement.
\end{proof}

From \Cref{thm:counterexample-construction}, it follows immediately that for all $\epsilon > 0$, a decision sequence satisfying $\epsilon$-Stronger PJR may not exist (as defined in \Cref{sec:hardness}), so \Cref{thm:counterexample} from the main body is proven.

We can use the same counterexample construction to show that some other related axioms are impossible to satisfy in general. In particular, a natural idea of defining PJR for our setting is by requiring that a group large enough to deserve $\ell$ decisions only needs to agree in $\ell$ rounds (instead of in $T$ rounds). This corresponds perhaps most closely to the original definition of PJR for multi-winner voting (which, in the context of electing $k$ candidates, requires agreement on $\ell$ candidates not $k$ candidates). It also corresponds to the definition of PJR proposed by \citet{freeman2020proportionality} who study multi-winner voting with a variable number of winners, which is a special case of our model where in each round, only 2 alternatives are available (``add $c$ to the committee'', ``do not add $c$ to the committee'').  However, this version of the PJR axiom cannot always be satisfied.

\begin{corollary}[$\ell$-Agreement PJR  need not exist]
	\label{cor:multi-winner-pjr}
	There exists an instance where no decision sequence $D$ satisfies ``\emph{$\ell$-Agreement PJR}'', defined to require that for every $\ell \in \mathbb N$ and every group of voters $S \subseteq N$ that agrees in at least $\ell$ rounds and that has size $|S| \geq \ell\cdot \frac{n}{T}$, there are at least $\ell$ rounds $j$ with $d_j \in \bigcup_{i \in S} \smash{A_j^i}$.
\end{corollary}
\begin{proof}
	Invoke \Cref{thm:counterexample-construction} with $\epsilon=0.5$, $k = 4$, and $T=40$. The theorem gives an instance where for every decision sequence $D$, there is a group of voters $S \subseteq N$ which agrees in $3$ rounds, has size $|S| \ge (1-\epsilon)\cdot \frac{n}{k} = 0.125\cdot n \ge 0.1\cdot n = 4 \cdot \frac{n}{T}$, but none of the voters in $S$ approve any of the decisions in $D$, violating $\ell$-Agreement PJR.
\end{proof}

When there are only 2 candidates in every round, Sequential Phragm\'en, MES, and PAV all satisfy this property \citep{freeman2020proportionality}.

Following \citet{DoHL022}, we now consider two relaxations of the $\ell$-Agreement PJR axiom which introduce multiplicative approximations by a factor $\alpha$. We show that for each constant $0 < \alpha < 1$, there are instances where even this relaxed version is not satisfiable. 

The first relaxation requires a group that the group approves only $\lfloor \alpha \cdot \ell \rfloor$ decisions instead of $\ell$ decisions.

\begin{corollary}[$\alpha$-$\ell$-Agreement PJR need not exist]
	Let $0 < \alpha < 1$. There exists an instance where no decision sequence $D$ satisfies  ``\emph{$\alpha$-$\ell$-Agreement PJR}'',
	defined to require that for every $\ell \in \mathbb N$ and every group of voters $S \subseteq N$ that agrees in at least $\ell$ rounds and that has size $|S| \geq \ell \cdot \frac{n}{T}$, there are at least $\lfloor \alpha \cdot \ell \rfloor$ rounds $j \in R$ with $d_j \in \bigcup_{i \in S} \smash{A_j^i}$.
\end{corollary}
\begin{proof}
	Invoke \Cref{thm:counterexample-construction} with $\epsilon=0.5$, $k = \lceil \frac{1}{\alpha} \rceil$, and $T = \lceil \frac{8}{\alpha^2} \rceil$ to obtain an instance.
	
	Assume that $D$ is a decision sequence satisfying $\alpha$-$\ell$-Agreement PJR. The theorem says that there exists a group $S$ which agrees in $k$ rounds and has size $|S| \ge 0.5 \cdot \frac{n}{k}$ and such that none of the decisions in $D$ are approved by any member of $S$.
	
	Write $\ell = k$. Note that $S$ agrees in $\ell$ rounds and (using the fact that $\lceil x \rceil \le 2x$ for all $x > 1$) that
	\[
		|S| 
		\ge 0.5 \cdot \frac{n}{k} 
		= 0.5 \cdot \frac{n}{\lceil 1/\alpha \rceil} 
		\ge 0.5 \cdot \frac{n}{2/\alpha} 
		= \frac{1}{4} \cdot \frac{\alpha^2}{\alpha} \cdot n 
		= \frac{2}{\alpha} \cdot \frac{n}{8/\alpha^2}
		\ge \ell \cdot \frac{n}{T}.
	\]
	Hence, $\alpha$-$\ell$-Agreement PJR requires that in at least $\lfloor \alpha \cdot \ell \rfloor = \lfloor \alpha \cdot \lceil \frac{1}{\alpha} \rceil \rfloor \ge \lfloor \alpha \cdot \frac{1}{\alpha} \rfloor = 1$ rounds, at least one member of $S$ approves the decision of $D$, contradiction.
\end{proof}

The second relaxation requires that the coalition agrees in $\ell/\alpha$ rules to be guaranteed to approve $\ell$ decisions.

\begin{corollary}[$\ell/\alpha$-Agreement-PJR need not exist]
	Let $0 < \alpha < 1$. There exists an instance where no decision sequence $D$ satisfies  ``\emph{$\ell/\alpha$-Agreement-PJR}'',
	defined to require that for every $\ell \in \mathbb N$ and every group of voters $S \subseteq N$ that agrees in at least $\ell/\alpha$ rounds and that has size $|S| \geq \ell \cdot \frac{n}{T}$, there are at least $\ell$ rounds $j$ with $d_j \in \bigcup_{i \in S} \smash{A_j^i}$.
\end{corollary}
\begin{proof}
	Invoke \Cref{thm:counterexample-construction} with $\epsilon=0.5$, $k = \lceil\frac{1}{\alpha}\rceil$, and $T=\lceil\frac{4}{\alpha}\rceil$ to obtain an instance. 
	
	Assume that $D$ is a decision sequence satisfying $\ell/\alpha$-Agreement-PJR. The theorem says that there exists a group $S$ which agrees in $k$ rounds and has size $|S| \ge 0.5 \cdot \frac{n}{k}$ and such that none of the decisions in $D$ are approved by any member of $S$.
	
	Write $\ell = 1$. Note that $S$ agrees in $k = \lceil\frac{1}{\alpha}\rceil \ge \frac{\ell}{\alpha}$ rounds and that
	\[
	|S| 
	\ge 0.5 \cdot \frac{n}{k} 
	\ge 0.5 \cdot \frac{n}{2/\alpha} 
	= 0.25 \cdot \alpha \cdot n 
	= \frac{n}{4/\alpha} 
	\ge \ell \cdot \frac{n}{T}.
	\]
	Hence $\ell/\alpha$-Agreement PJR requires that in at least $\ell = 1$ rounds, at least one member of $S$ approves the decision of $D$, contradiction.
\end{proof}

We have seen several axioms that cannot always be satisfied. However, we could hope to find rules that satisfy these axioms whenever they are run on an instance where there exists a decision sequence satisfying this axiom. While there are certainly artificial \emph{offline} rules doing this (e.g., the rule that outputs an arbitrary sequence satisfying the axiom should one exists, and otherwise outputs the PAV sequence), it is not clear if there are natural such rules. 

Online or semi-online rules, however, provably cannot satisfy such a property, as one can see by adapting the counterexample construction of \Cref{thm:counterexample-construction}:
The first $T - 1$ rounds remain exactly the same as in that construction. We observe the decisions made by the online rule in these rules. In the last round, take $s+1$ voters who do not approve any of the decisions made thus far, and have each of them approve a distinct alternative (like in the original construction). For the remaining $n - (s+1)$ voters, introduce an alternative $c$ that they all approve. No matter which alternative the rule chooses in round $T$, there remain $s$ voters who do not approve any decision, leading to a violation. However, a violation could have been avoided by choosing in round 1 an alternative commonly approved by $s$ of the $s+1$ voters, and choosing $c$ in round $T$. Thus, we have arrived at a situation where a violation-free decision sequence exists but an online (or semi-online) rule cannot find it as we alter the input based on its past decisions. This proves \Cref{thm:counterexample-online} from the main body. This kind of construction works for all the axiom variants we have considered in this section.

\subsection{Can Online Rules Satisfy Weak EJR?}
\label{app:online-ejr}

We have seen that the online rule Sequential Phragm\'en fails Weak EJR, while the semi-online MES satisfies Weak EJR. Is it possible for an online rule to satisfy EJR or Weak EJR?
For EJR, this question has been answered in the negative for EJR by \citet{elkind2024verifying}, even for semi-online rules.
For Weak EJR, we do not know the answer to this question.\footnote{In the special case where the set of available alternatives and the approval sets are the same across rounds, a known online method satisfies EJR \citep[Theorem 4.3]{approvalbasedapportionment}. Sequential Phragm\'en fails Weak EJR even in this special case (\Cref{thm:phragmen_fails_EJR}).}
However, we show that any online rule that satisfies Weak EJR for sequential decision making can be translated into a multi-winner voting rule that satisfies EJR and committee monotonicity.
As \citet[Sec.~7.1]{lacknerskowron} write, a ``main open question is whether there exist [approval-based multi-winner] rules that satisfy EJR and committee monotonicity''. Thus, we leave finding an online rule that satisfies (Weak) EJR as a challenging problem for future work.

In the remainder of this section, we describe the way in which an online decision rule can be used as a multi-winner voting rule. To do so, we must first give basic definitions about the latter topic. 

\begin{definition}[Multi-winner Voting Rule] 
Given a set of voters $N$, a set of alternatives $C$, a profile of approval sets $A = (A^i)_{i\in N}$ (with $A^i \subseteq C$ for all $i \in N$), and a desired committee size $k$, a \emph{multi-winner voting rule} $f$ outputs a committee $f(N, C, A, k) = W \subseteq C$ with $|W| = k$.
\end{definition}

Intuitively, a multi-winner voting rule is said to be committee monotonic if it selects the winning committee for size $k$ by first computing the winning committee for size $k-1$, and then adding a $k$th alternative. Thus, when an alternative is a member of the winning committee, it can never become a losing alternative when the committee size $k$ is increased.

\begin{definition}[Committee Monotonicity] 
A multi-winner voting rule $f$ satisfies \emph{committee monotonicity} if for all $N$, $A$, $C$, and $k \ge 1$, we have $f(N, C, A, k-1) \subset f(N, C, A, k)$.
\end{definition}

We refer to the EJR axiom originally proposed for the multi-winner voting context by \citet{ejr} as Multi-Winner EJR to distinguish it from our axioms for the sequential context.

\begin{definition}[Multi-Winner EJR] 
Let $k$ be the desired committee size. A committee $W$ satisfies \emph{Multi-Winner EJR} if for every $\ell \in \mathbb{N}$, for every group of voters $S \subseteq N$ of size $|S| \geq \ell \cdot \frac{n}{k}$ with $\smash{|\bigcap_{i\in S} A^i|} \ge \ell$, there is at least one voter $i \in S$ with $\smash{|W \cap A^i|} \ge \ell$.
\end{definition}

We can now formally state our reduction.

\begin{theorem}
    An online decision rule that satisfies Weak EJR in sequential decision making induces a multi-winner voting rule that simultaneously satisfies Multi-Winner EJR and committee monotonicity. 
\end{theorem}

\begin{proof}
Let $f$ be an online decision rule that satisfies Weak EJR. We will derive a multi-winner voting rule by describing a procedure that, given a profile of approval sets, selects alternatives in some order one-by-one until all alternatives have been selected. Let us label the alternatives in order of selection as $d_1, \dots, d_m$. To obtain a committee of size $k$, we will take $W = \{d_1, \dots, d_k\}$. By construction, such a method must be committee monotonic.

So let $N$ be a set of voters, $C$ be a set of alternatives, and $A = (A^i)_{i \in N}$ be a profile of approval sets over $C$.
We use our online rule $f$ in the following adaptively constructed sequence of $T = |C|$ rounds.
In round $1$, the set of available alternatives is $C_1 = C$, and the approval sets are $A_1^i = A^i$ for every $i \in N$. We apply $f$ to this round and thereby obtain $d_1$, the first selected alternative.

Thereafter, in round $j = 2, \dots, m$, we have already obtained a sequence of decisions $d_1, \dots, d_{j-1}$ which we will assume inductively are pairwise distinct. In round $j \in R$, we take the set of available alternatives to be $C_j = C \setminus \{d_1, \dots, d_{j-1}\}$. Further, we take the approval sets to be $A_j^i = A^i \cap C_j$ for all $i \in N$. Then we apply our rule $f$ to obtain the next decision $d_j$, which by choice of $C_j$ is distinct from all prior decisions.

To finish, we need to prove that the multi-winner voting rule we have defined satisfies Multi-Winner EJR. Let $k$ be a desired committee size, and let $W = \{d_1, \dots, d_k\}$. Let $\ell \in \mathbb N$ and suppose that $S\subseteq N$ is a group of voters with $|S| \geq \ell \cdot \frac{n}{k}$ and $\smash{|\bigcap_{i\in S} A^i|} \ge \ell$. Assume for a contradiction that every member of $S$ approves at most $\ell - 1$ alternatives in $W$. Then there must exist an alternative $a$ that all members of $S$ approve, but that is not a winner: $a \in \bigcap_{i\in S} A^i \setminus W$. Now consider the decision instance given by the first $k$ rounds constructed above (thus, having time horizon $T = k$). For that decision instance, note that $S$ forms a group of voters that agrees in every round (because $a \in C_j$ for all $j = 1, \dots, k$). Also $|S| \ge \ell \cdot \frac{n}{T}$. Hence, because $f$ satisfies Weak EJR, there is a voter $i \in S$ who approves at least $\ell$ of the decisions made by $f$ on this decision instance, a contradiction.
\end{proof}

\section{Full Results from Experiments}
\label{sec:fullresults}

\subsection{Learning Preferences from the Moral Machine}
\label{app:moral-machine-results}

\begin{figure*}[t]
	\centering
	\begin{subfigure}{\textwidth}
		\centering
		\includegraphics[width=0.82\textwidth]{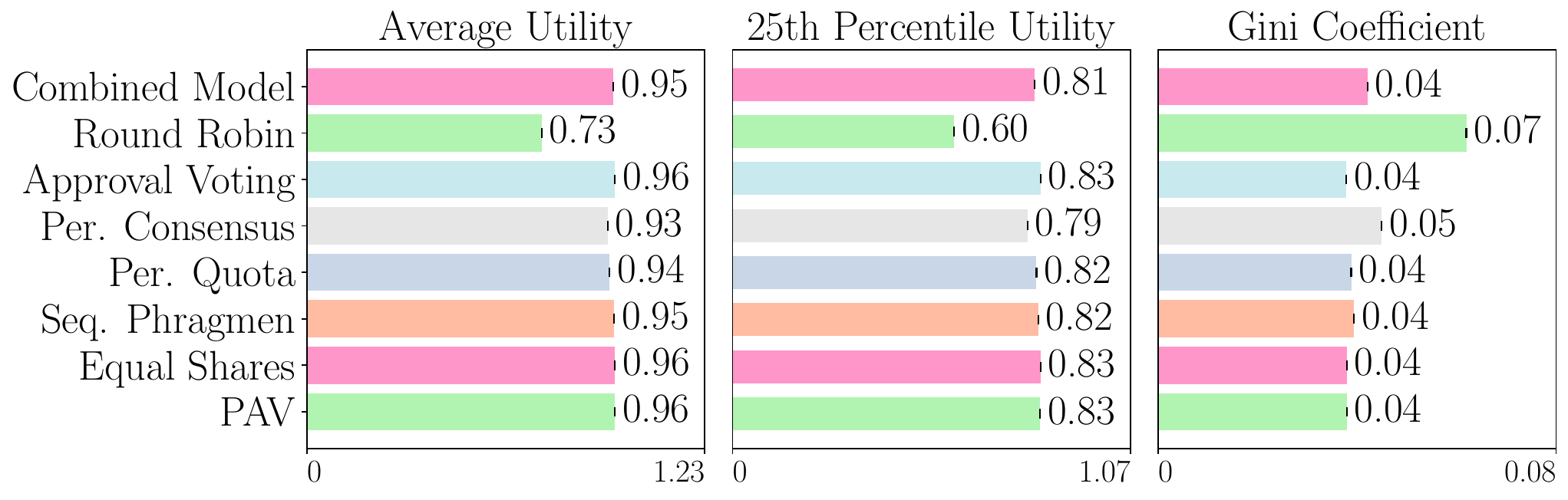}\caption{$100$ rounds with each round having $50$ alternatives and each voter approving all alternatives having more than average utility.}
		\label{fig:50_approving_above_abg_random}
	\end{subfigure}

	\vspace{\baselineskip}
	
	\begin{subfigure}{\textwidth}
		\centering
		\includegraphics[width=0.82\textwidth]{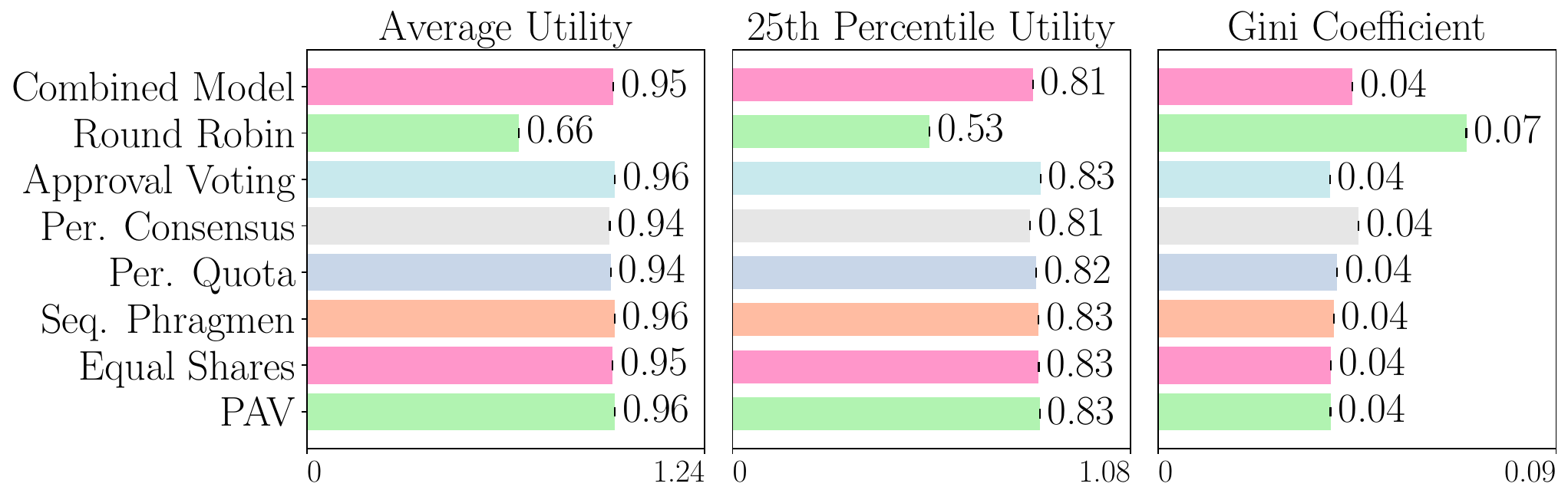}
		\caption{$100$ rounds with each round having $50$ alternatives and each voter approving top 2 alternatives. }
		\label{fig:50_approving_top2_random}
	\end{subfigure}
	
	\vspace{\baselineskip}
	
	\begin{subfigure}{\textwidth}
		\centering
		\includegraphics[width=0.82\textwidth]{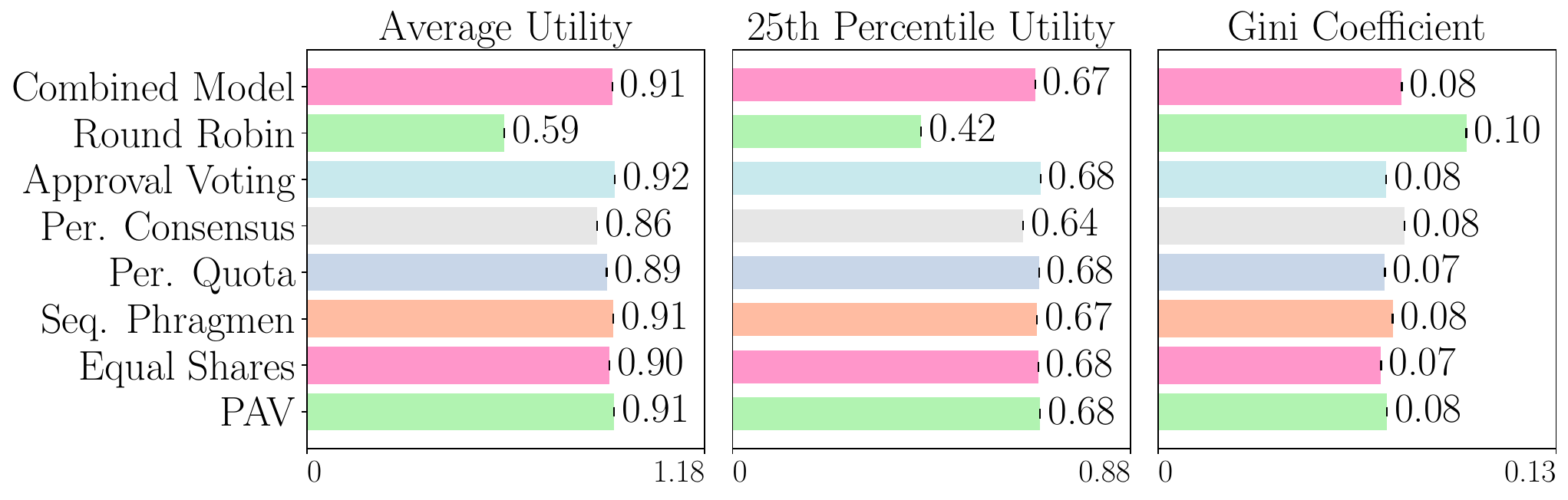}\caption{$100$ rounds with each round having $100$ alternatives and each voter approving top $1$ alternatives.}
		\label{fig:50_appoving_top10_random}
	\end{subfigure}
	
	\vspace{\baselineskip}
	
	\begin{subfigure}{\textwidth}
		\centering
		\includegraphics[width=0.82\textwidth]{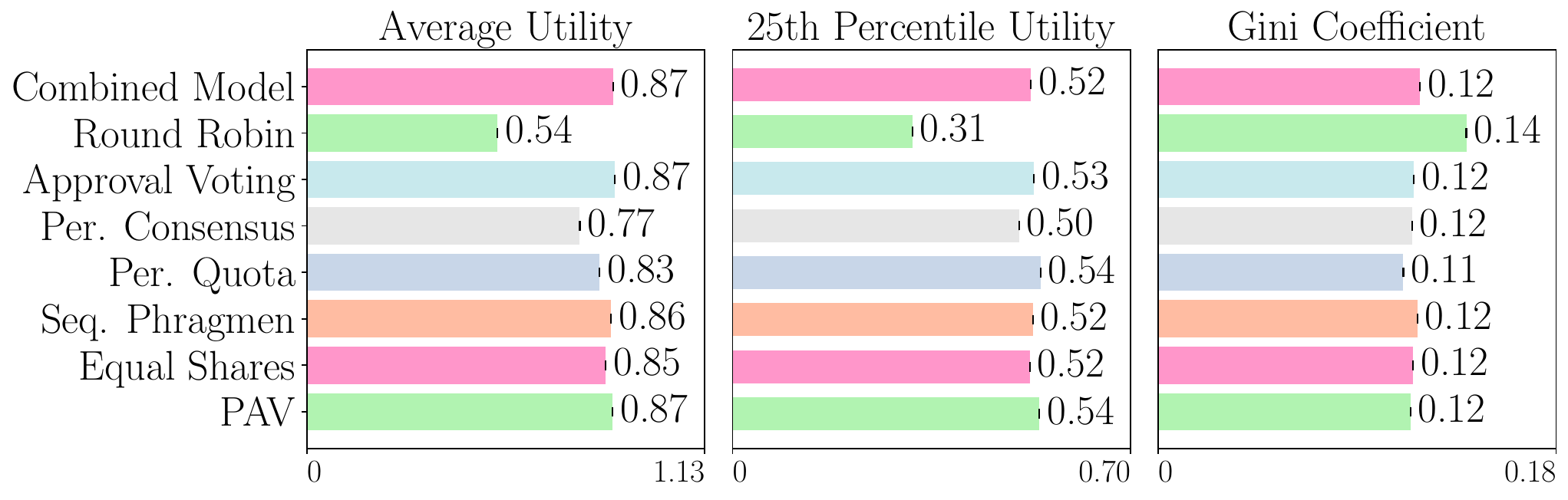}\caption{$100$ rounds with each round having $200$ alternatives and each voter approving top $2$ alternatives.}
		\label{fig:200_appoving_top2_random}
	\end{subfigure}
	
	\caption{Performance of the different rules on the Moral Machine Dataset based on \emph{random} alternatives.}
	\label{fig:moral_machine_agreement}
\end{figure*}

\begin{figure*}[t]
    \centering
    
    \begin{subfigure}{\textwidth}
        \centering
        \includegraphics[width=0.82\textwidth]{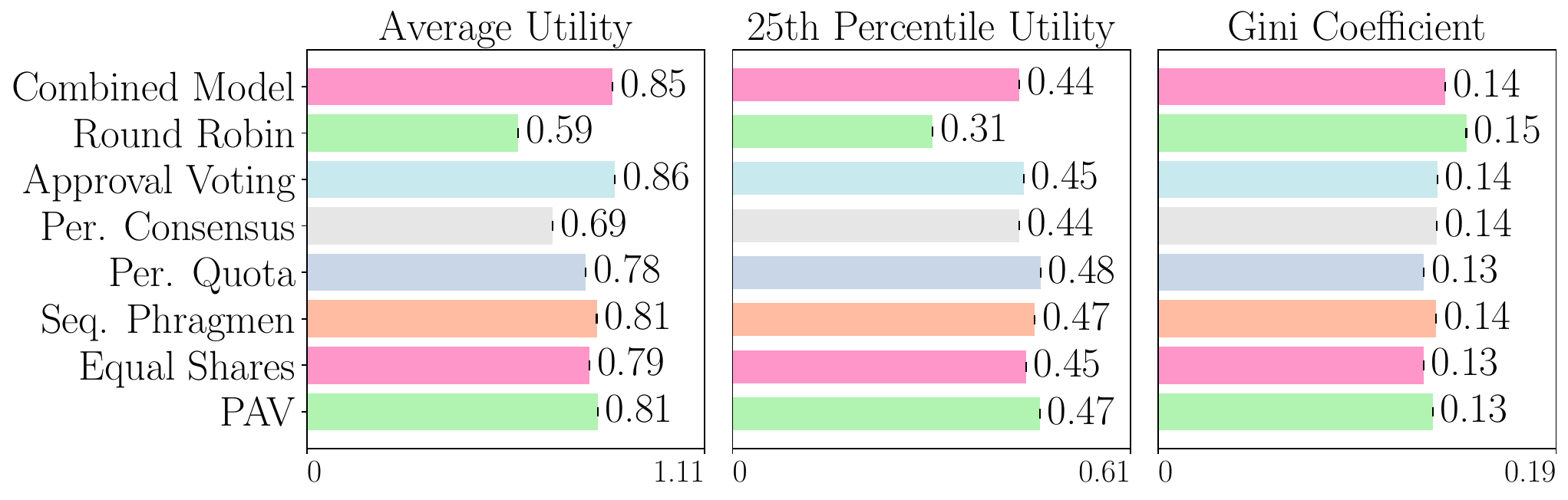}\caption{$100$ rounds with each round having $50$ alternatives and each voter approving all alternatives having more than average utility.}
        \label{fig:50_approving_above_abg}
    \end{subfigure}

	\vspace{\baselineskip}
    
    \begin{subfigure}{\textwidth}
        \centering
        \includegraphics[width=0.82\textwidth]{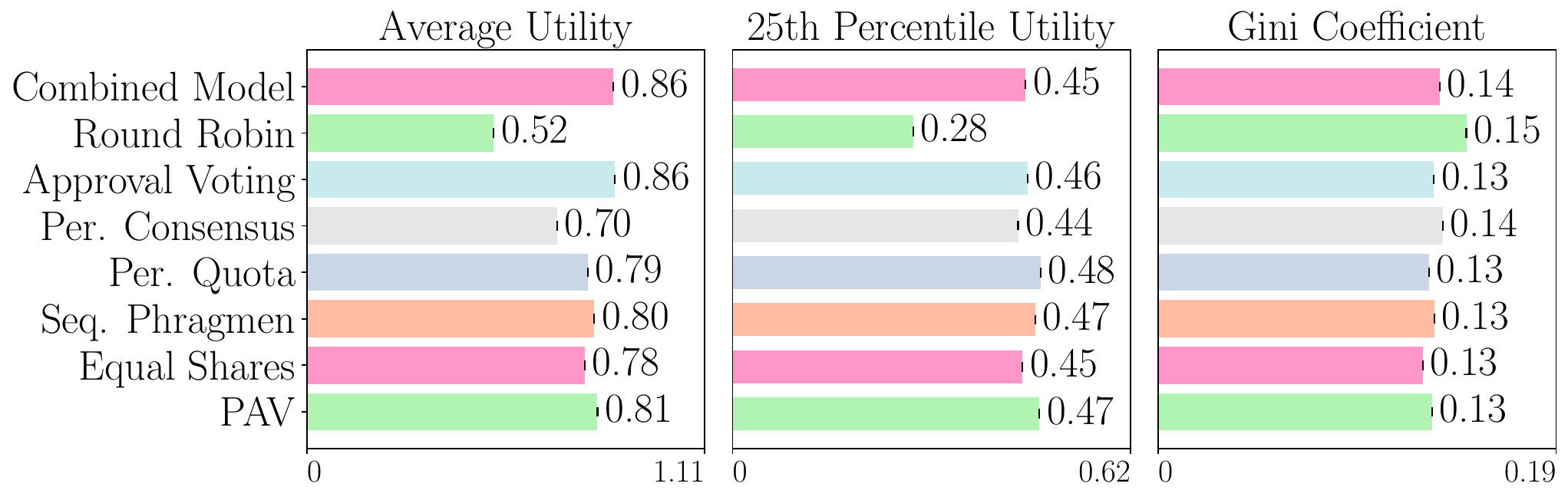}
        \caption{$100$ rounds with each round having $50$ alternatives and each voter approving top 2 alternatives. }
        \label{fig:50_approving_top2}
    \end{subfigure}
    
    \vspace{\baselineskip}
    
    \begin{subfigure}{\textwidth}
        \centering
        \includegraphics[width=0.82\textwidth]{Figures/moral_machine/disagreement_balanced_model_100_rounds=100_altPerRound=100_approving_top_1_Trials=1000.pdf}\caption{$100$ rounds with each round having $100$ alternatives and each voter approving top $1$ alternatives.}
        \label{fig:50_appoving_top10}
    \end{subfigure}
    
    \vspace{\baselineskip}
    
    \begin{subfigure}{\textwidth}
        \centering
        \includegraphics[width=0.82\textwidth]{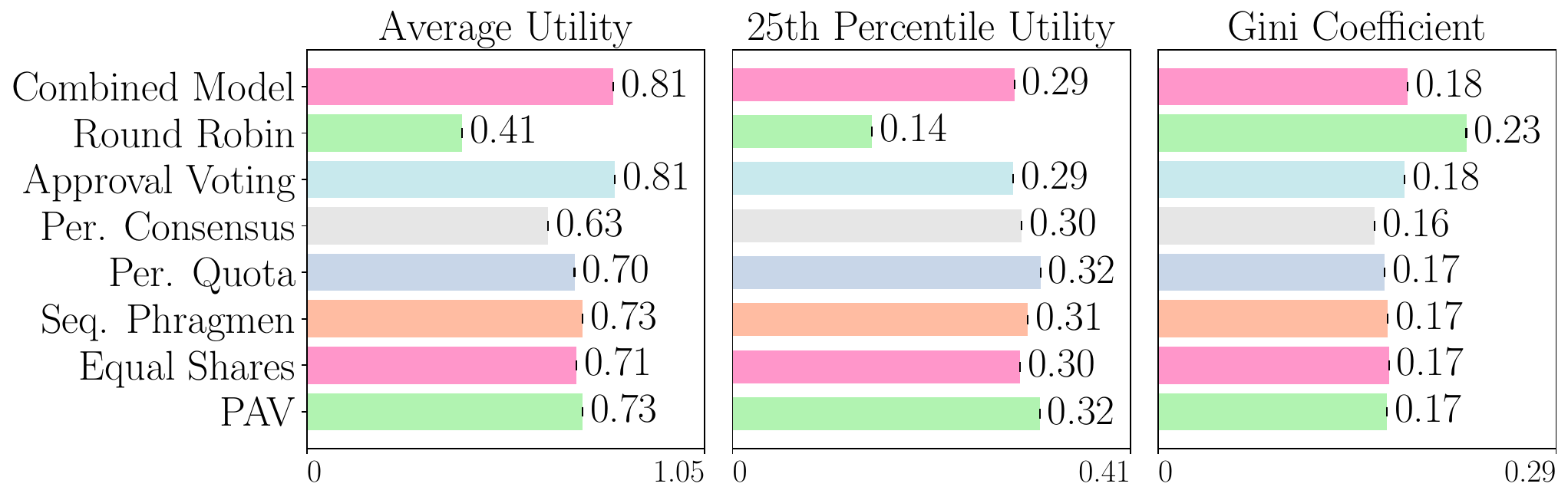}\caption{$100$ rounds with each round having $200$ alternatives and each voter approving top $2$ alternatives.}
        \label{fig:200_appoving_top2}
    \end{subfigure}

 \caption{Performance of the different rules on the Moral Machine Dataset, based on \emph{high-disagreement} alternatives.}
    \label{fig:moral_machine_disagreement}
\end{figure*}
  
In this appendix, we present results for some alternative parameter settings. We will vary the number of alternatives and the sizes of voters' approval sets.
Indeed, we consider four different parameter settings (a)-(d). We produce 100 decision rounds. In each round, we produce either 50 alternatives ((a)-(b)), 100 alternatives ((c)) or 200 alternatives ((d)), in one of the following ways:
\begin{itemize}
	\item \emph{Random.} We create a large pool of 240\,000 unseen alternatives and sample a random subset of these as the alternatives for each round. 
	\item \emph{High-disagreement alternatives.} We find that on random alternatives, the country models have high agreement (i.e., they tend to approve the same alternatives from the sampled set of alternatives), making the aggregation problem relatively straightforward. Thus, a more interesting case involves finding alternatives with high disagreement. We quantify disagreement on an alternative as high variance in the ranks assigned by each model to the alternative. From the pool of all 240\,000 alternatives present in the dataset, we pick the 500 alternatives with the highest variance in ranks. Then in each round we sample alternatives randomly from this high-disagreement set of 500 alternatives.
\end{itemize}
Finally, we generate approval sets by computing the utilities for each alternative as predicted by the country model, and let the country approve (a) the alternatives with above-average utility, (b)/(d) the top 2 alternatives, (c) the top 1 alternatives

As before, we use our voting rules to compute decision outcomes using the country models as voters. 
To compare the voting rules to the performance of the combined model, we pick the alternative assigned the highest utility by the combined model as its decision for the round.
We present the results for \emph{random alternatives} in \Cref{fig:moral_machine_agreement}
and for \emph{high-disagreement alternatives} in \Cref{fig:moral_machine_disagreement}.

\clearpage 
\bibliographystyle{ACM-Reference-Format}
%\bibliography{perpetual-proportional-jair}

\end{document}